\documentclass{scrartcl}
 \usepackage{cmap}
 \usepackage{amsmath,amsxtra,amscd,amssymb,latexsym,stmaryrd,amsthm, mathtools,mathrsfs}
  \usepackage{xcolor,graphicx}

\usepackage{lmodern} 
\usepackage[english]{babel}
\usepackage[utf8]{inputenc}
\usepackage[T1]{fontenc}

\usepackage[font=small]{caption}

\usepackage[colorlinks=true,linkcolor=red!70!black,citecolor=green!70!black, urlcolor=magenta!70!black,backref]{hyperref}
\usepackage{enumitem,url}

\addtokomafont{title}{\rmfamily\itshape}

\theoremstyle{plain}
\newtheorem{theo}{Theorem}[section]
\newtheorem*{theo*}{Theorem}
\newtheorem{coro}[theo]{Corollary}

\newtheorem{lemm}[theo]{Lemma}
\newtheorem{theomain}{Theorem}

\theoremstyle{definition}
\newtheorem{defi}[theo]{Definition}
\newtheorem{rema}[theo]{Remark}
\newtheorem{exem}[theo]{Example}

\newenvironment{important}{\begin{center}\begin{minipage}{.8\linewidth}\textit\bgroup}{\egroup\end{minipage}\end{center}}

\newcommand*{\dd}%
  {\relax\ifnum\lastnodetype>0\mskip\medmuskip\fi\mathrm{d}}

\newcommand{\limit}{\mathsf{limit}}

\usepackage{diagbox}

\usepackage{caption}
\usepackage{subcaption}

\newcommand{\pref}{\succsim}
\newcommand{\spref}{\succ}

\title{Cycles in synchronous iterative voting: general robustness and examples in Approval Voting}

\author{Beno\^{\i}t R. Kloeckner,\footnote{ORCID 0000-0002-4966-7864, \url{benoit.kloeckner@u-pec.fr}} \thanks{Univ Paris Est Creteil, CNRS, LAMA, F-94010 Creteil, France -- Univ Gustave Eiffel, LAMA, F-77447 Marne-la-Vallée, France}}

\usepackage[section]{placeins}

\begin{document}

\maketitle

\begin{abstract}
\noindent\textbf{Abstract.} We consider synchronous iterative voting, where voters are given the opportunity to strategically choose their ballots depending on the outcome deduced from the previous collective choices.

We propose two settings for synchronous iterative voting, one of classical flavor with a discrete space of states, and a more general continuous-space setting extending the first one. We give a general robustness result for cycles not relying on a tie-breaking rule, showing that they persist under small enough perturbations of the behavior of voters. Then we give examples in Approval Voting of electorates applying simple, sincere and consistent heuristics (namely Laslier's Leader Rule or a modification of it) leading to cycles with bad outcomes, either not electing an existing Condorcet winner, or possibly electing a candidate ranked last by a majority of voters. Using the robustness result, it follows that those ``bad cycles'' persist even if only a (large enough) fraction of the electorate updates its choice of ballot at each iteration.

We complete these results with examples in other voting methods, including ranking methods satisfying the Condorcet criterion; an \emph{in silico} experimental study of the rarity of preference profiles exhibiting bad cycles; and an example exhibiting chaotic behavior.
\end{abstract}

\section{Introduction}
 
\paragraph*{Iterative Voting.}

In any voting system, the choice of a ballot by any voter is guided by her preferences between the candidates and the expected effect of each ballot she can cast.
To estimate the effect of any ballot, one needs to know the ballots cast by the other voters, or at least to have some information on them. Assuming perfect information is neither realistic nor theoretically useful: in large electorates, changing one's ballot rarely changes the outcome of the election (hence, as is well-known, the notion of Nash equilibrium is ineffectively broad in the context of elections \cite{Laslier2009leader}). One thus often considers situations where information is imperfect, so that the ballot to be chosen by any given voter might change the outcome, with a tiny but non-zero probability. Another point of view is to consider the possibility for voters of aligned interests to form coalitions and decide together which ballots to cast. A third point of view, taking from both of the first two, is to consider that some (or all) information about voters' intents is common knowledge, and that voters choose their ballot under the assumption that other voters that have the same preferences as theirs will make the same reasoning as they will -- thus overcoming the insufficient weight of one ballot to change the outcome. 

In any case, these considerations introduce a \emph{dynamical} component to voting: after all voters adjusted their intentions, the information under which this adjustment was made is outdated; the new intents result in new information that could be shared, and itself result in new strategic adjustments of voters' intents, etc. Given a model of voter's response to new information, one obtains a new set of equilibria, which we call \emph{dynamical equilibria}: situations when the intents of voters stay unchanged after the outcome they induce is broadcast. In this article, the single word \emph{equilibrium} will only be used to denote dynamical equilibria, never Nash equilibria. Natural questions are thus whether such dynamical equilibrium exist; whether such equilibrium is unique; what properties such equilibriums have, in particular which candidates can be elected at equilibrium; and which other sets of intents converge to an equilibrium after successive adjustments. The field of \emph{iterative voting} is the study of these questions, under various modeling assumptions.

\paragraph*{Synchronized and asynchronous iterative voting.}

By far the most studied version of iterative voting is \emph{asynchronous}, i.e. only some voters (most commonly a single one) adjust their intents given the available information, and information is updated before some other voter(s) make their adjustments, etc. While this model has both a theoretical interest and practical applications to some collective decision processes, it is ill-suited to model large-scale political election, where the number of times when the information is updated (poll publications) is much lesser than the number of voters.\\

In contrast, we consider in this article \emph{synchronous} iterative voting, where all voters are given the information at time $t$, and simultaneously adjust their intents to build the new situation at time $t+1$. 
It is important to observe that what makes iterative voting synchronous is not that the voters' intents change exactly at the same time, but that all voters have the opportunity to adjust their preferences before information is updated.

Synchronous iterative voting can be thought of as a model for political elections, where information is broadcasted through polls; we will thus name our dynamical model \emph{Polling Dynamics}; similar points of view have been notably taken in \cite{chopra2004knowledge}, and \cite{reijngoud2012response} where the emphasis is on the effect of the amount of information given by the polls (ranks, or scores, etc). 

However this is not our sole motivation: considering synchronized iterative voting is also interesting from a theoretical perspective, when trying to define a ``best'' ballot for each voter. Indeed, to be efficient in her choice a voter should not only vote strategically, but also anticipate that other voters will also do so, who will themselves try to anticipate all voters' strategies, etc. A model for this situation is to replace instantaneous anticipations by iterations: each voter is given a fixed heuristic, which dictates for each possible expected outcome which ballot to cast, and we inductively apply these heuristics and update the expected outcome. At a fixed point of the Polling Dynamics, which we call a \emph{dynamical equilibrium}, anticipating each other's strategy result in the same situation, which is therefore stable. Conversely, the existence of a cycle means that in some circumstances there is no meaningful way for voters to inductively anticipate other's strategies. Note that even when limiting the number of steps of counter-strategies, the game-theoretic analysis of voting rules can be quite deep, see \cite{grandi2019GSgames}.\\

We shall be especially interested in Approval Voting -- the voting system in which a ballot can contain the names of any subset of candidates, and the candidate whose name is present in the most ballots is elected. Approval Voting has the very interesting feature that several sincere ballots coexist for each order of preference, so that there exist heuristics that are both sincere, simple, and non-trivial; let us discuss this in more detail.

\paragraph*{On sincerity, strategic voting and straightforwardness} 

Durand notes in \cite{Durand2015manipulables} that in general the meaning of ``sincerity'' is open to interpretation, and that this word has often been used to argue against strategic voting. He makes a compelling point that strategic voting is to be expected, and even advised to voters, and that what causes a democratic problem is the \emph{necessity} of resorting to strategies to get the best outcome rather than the fact that voters embrace this necessity (``\emph{Manipuler c'est bien, la manipulabilité c'est mal}'': ``Manipulation is good, manipulability is bad''). Indeed strategies that either need a lot of information on others' plans or that are too contrived cannot be applied uniformly, creating asymmetries between voters. Even in the absence of manipulation or strategic voting, sincere ballots can be cast but afterward regretted in view of the outcome of the election, thus lowering trust and confidence in the democratic system. A most important property is thus \emph {straightforwardness} \cite{gibbard1973manipulation}, i.e. the existence of strategies that yield an optimal outcome whatever choices are made by the other voters. Since straightforwardness cannot be hoped in general \cite{gibbard1973manipulation,satterthwaite1975strategy}, one can try to determine whether certain systems often offer weak forms of straightforwardness (see again \cite{Durand2015manipulables}).

In Approval Voting, one can say that a ballot is \emph{sincere} whenever all candidates present on the ballot are preferred to all candidates not present on it \cite{Brams2007approval}. As soon as there are more than two candidates each voter has several sincere ballots, corresponding to the various points in her order of preference where she can draw the line between acceptance and rejection. In particular, strategic voting (i.e. choosing one's ballot depending on the ballots expected  to be cast by the other voters) can occur even when restricting to sincere ballots. In other words, sincerity in the above sense does not imply straightforwardness. One of our goals is to give examples showing how very far from straightforward Approval Voting can be even when restricting to sincere ballots.

\paragraph*{Is Approval Voting a Condorcet system ``in practice''?} 

While Approval Voting is known not to be a Condorcet system, it has been argued that it might be close to Condorcet ``in practice''. 
We can e.g. read in \cite{brams2010going}: ``In particular, Condorcet candidates [...] almost always win under AV'', and some elections where this assertion can be checked are mentioned as evidence. 

Several theoretical results could also be seen as providing a basis for the above claim. First, even though many other voting systems also satisfy these properties, let us mention that Brams and Sanver \cite{Brams2003sovereignty} showed that when a Condorcet winner exists, her election is a strong Nash equilibrium; and in \cite{sertel2004strong} it is shown that only Condorcet Winners are elected at strong Nash equilibriums. Strong Nash equilibrium are said by Brams and Sanver to be ``globally stable'', but one point not considered in either works above is whether they should be expected to be reached in practice.

A stronger argument originates in Laslier \cite{Laslier2009leader}, where it is proved that under a large-electorate model with uncertainty in the recording of votes and perfect common information, the best course of action for voters results in a particular heuristic, the ``Leader Rule'' (LR). Additionally, he proved that if there is a Condorcet winner and all voters apply the Leader Rule, then there is at least one equilibrium, and any equilibrium elects the Condorcet winner.\footnote{A similar argument has been raised by Warren D. Smith and is used by advocates of Approval Voting and Range Voting, see https://www.rangevoting.org/AppCW.html. The argument of Smith is less rigorous than Laslier's, since the equilibrium assumption is kept implicit.}

\paragraph*{Description of the main results.} 

To introduce the Polling Dynamics, we propose both a classical-flavored \emph{discrete space} setting, and an enlarged \emph{continuous space} setting that allows to continuously perturb a situation (i.e. a given electorate with fixed heuristics), for example in order to assume that at each step of the dynamics, a small fraction of voters fail to adjust their ballots. This makes it possible to discuss how dynamical phenomenons change under such perturbations. 

Our first main result (Theorem \ref{t:robust} in Section \ref{s:robust}) is a quite general robustness result, showing under a mild assumption that cycles or equilibriums in the discrete space setting persist under perturbation inside the continuous space setting. All the examples we will give benefit this robustness. Our main contribution here is in the framework for Polling Dynamics, the proof of Theorem \ref{t:robust} being rather straightforward once the definitions are set up.

Among the possible perturbations, one can model the situation where instead of having all voters change their ballots when informed of the last poll, a (small enough) fraction keeps their ballot as they were. This particular case brings us closer to asynchronous voting, but only to some extent since we need to stay close to the discrete, synchronous model.
\\

The next main results are examples showing that in Approval Voting, the Polling Dynamics can exhibit a ``bad cycle'' where a disliked candidate can be elected, even in the presence of a Condorcet winner (Theorems \ref{t:LR} and \ref{t:example2} in Section \ref{s:AVexamples}). The main point shall not be the mere existence of cycles, a rather unsurprizing phenomenon,\footnote{An example of cycle under the Leader Rule is given as exercise 8.4.4 in \cite{meir2018strategic}, but there are no Condorcet winner.} but rather that such cycles can result in the election of a suboptimal candidate. As noted by Laslier, previous examples of cycles (notably in \cite{Brams2007approval}) needed some voters to change their heuristic at some iteration of the process; in our examples the assumed strategic behaviors is consistent, i.e. constant in time; they are also sincere, simple, and strategically sound (the first example uses the Leader Rule, the second a slight modification where voters refuse to vote for several of their least-preferred candidates in any circumstance). Two obvious limitations to the relevance of these existence results are that such ``bad cycles'' could be rare, and that it is unrealistic to assume all voters to update their choice of ballots at each poll publication. The former limitation will be addressed by numerical experiments (which partly confirm it, and partly moderate it, see below); the latter is addressed by our robustness result. Indeed, Theorem \ref{t:robust} shows that if instead of assuming that all voters update their choice of ballot at each iteration, we only assume that a large enough proportion of them do, then the same cycle of winners still occurs.

Regarding whether Approval Voting should often elects the Condorcet winner, our examples cast a doubt. Laslier's results mentioned above mean that in every situation where a Condorcet winner exists but is not elected, voters (individually) have a strategical interest in changing their ballots. Our examples expose that this does \emph{not} mean that following this heuristic shall lead to the election of a Condorcet winner: it may also lead to a periodic pattern where the Condorcet winner is absent. More generally, even simple heuristics make the availability of information a possible impediment to reaching an equilibrium and/or to elect an existing Condorcet winner.

This new issue with the Approval Voting system complement in particular the examples provided in \cite{Sinopoli2006examples}.\\

%
%
%

We also provide a few examples in other voting systems in Section \ref{s:other}. Together with the Approval Voting examples, they show how extreme an impact polling can have on the election outcome: rigging \emph{any one} poll can prevent the election of the Condorcet winner even if \emph{all} subsequent polls are perfectly conducted and reported; and in fact, even if all polls are perfect but in the first one voters respond according to a pre-established expectation of the outcome of the election, this expectation can determine the outcome of the election even after arbitrary many polls: polling need not induce synchronization nor loss of memory. Even with perfect unrigged polls, the sheer \emph{number} of polls (e.g. its parity) can decide the outcome of the election.\\


It could be expected that such bad cycles are in practice rare. In Section \ref{s:experimental} we explore numerically this question, for Approval Voting, under various cultures modeling the construction of random voters' preferences. The main take-away is that bad cycles are rare when voters apply the Leader Rule, but can be more common, appearing in more than $15\%$ of preference profiles, in its modification where voters give themselves a limit to the candidates they may approve of. Most striking is that the culture where bad cycles are rarest for the Leader Rule is the one where they are most common for its modification.\\

The continuous-state setting enables one to model a broad variety of voter's collective behavior. In Section \ref{s:chaos}, we consider an example where voters are reluctant to approve of a candidate who is ranked low in their preference order, but may do so if a very close-call makes this move susceptible to improve the outcome. This arbitration can surprisingly make the Continuous-State Polling Dynamics a \emph{chaotic dynamical system}, making the prediction of the next winners from the observation of a sequence of previous expected winners practically impossible.

\paragraph*{Brief overview of the literature.}

Let us end this section by describing some more previous works. With no pretense to exhaustivity, the aim is to describe some directions previously explored in iterative voting; a more detailed review can be found in \cite{meir2017review}. The present work distinguishes itself from previous ones by considering synchronous iterative voting, previously seldomly studied: most commonly \emph{iterative voting} is used to mean what we named here \emph{asynchronous iterative voting}. Another divide which mostly parallels the previous one is between works considering at each step a set of possible moves by voters (e.g. best-responses, quite commonly considered in asynchronous voting) and, as we do here, deterministic heuristics. The reader can find an overview on strategic voting in \cite{meir2018strategic}.\\

In the case of Plurality Voting, discussed in slightly more details below, very general convergence results have been obtained in \cite{Meir2014dominance,Meir2015uncertainty,Meir2017iterative}. Many other voting rules --- not Approval --- have been considered in \cite{Lev2016convergence,Koolyk2017convergence}. Their theoretical results are negative (no guaranteed convergence) but empirical tests seem to indicate that cycles are rare. Restricted strategies, where voters may be constrained to some not necessarily best or better responses, have been studied in \cite{obraztsova2015restrictive} and \cite{grandi2013restricted}.

Non-myopic strategies have been considered for example in \cite{bowman2014separability,airiau2017learning,Obraztsova2015farsighted}, showing notably that voters learning from past information produce relatively good outcomes. Specific behaviors of voters, assumed to be subject to either \emph{truth bias} or \emph{lazy bias}, are studied in \cite{Rabinovich2015equilibria}.

Considering Plurality Voting, the situation is complicated by the rigidity of the single-name ballot, which forces voters to choose a trade-off between preferences and probability to improve the outcome of the election. The works \cite{Meir2014dominance} and \cite{Meir2015uncertainty} have studied in depth models taking into account the scores of the candidates and a level of uncertainty to define the possible voters' strategies. They obtained several results proving under some assumptions convergence to equilibrium (the result closest to our present setting is Theorem 5 in \cite{Meir2015uncertainty}, where at each iteration an arbitrary subset of voters adjust their votes according to the current poll results, thus including the case studied here where all voters adjust their votes at each iteration). In the case when strategies are restricted, \cite{obraztsova2015restrictive} gives sufficient conditions for convergence in many voting systems.

The question of the quality of equilibria have been considered for Plurality, Veto and Borda in \cite{branzei2013selfish}, where the voters can adjust their votes one at a time, and starting from the state where each voter casts her sincere ballot (which is unique in these voting systems). Concerning Plurality, they find that Nash equilibria that can be attained in this way are all very good; but in their model, the individual updates are made greedily and any candidate that would start with two or more votes less than the starting winner will never receive any new vote. One could expect real voters to vote not for immediate improvement of the winner, but in order to give a better position to a contender that might receive more support from others. This is the kind of modeling assumptions that our framework is made to support.

\paragraph*{Acknowledgments.}
I am indebted to Adrien Fauré \texttt{@AdrienGazouille} for  \href{https://twitter.com/AdrienGazouille/status/1181494450826014727}{a long debate on twitter} (in French) that lead me to seek and design the examples for Approval Voting presented here, and to François Durand for introducing me to the Social choice theory, for many discussions on Voting Systems, and for many suggestions that helped me improve this article. It also benefited from relevant comments provided by Adrien Fauré, Jean-François Laslier, Reshef Meir and anonymous referees who I warmly thank.

\section{Polling dynamics with discrete or continuous space}

This section mostly consists of notations and definitions; its main novelty is to propose a continuous space framework for iterative voting, which allows to model much more varied situations than the usual discrete  space versions. It models a situation where voters are grouped by type, e.g. according to their preference order, but where not all voters of a given type will react in the same way to an information update. For example, we can assume that a fraction of voters of each type is not aware of an information update, and will thus keep their previous intended ballot for the next iteration. In another direction, if the win is a close call, some voter's type may strategically adjust their ballots accordingly; but how close the call should be for them to bother changing their intent may vary from individual to individual, and our continuous space framework will make it possible to model this.

Examples of uses of the framework developed here are given below, first for Approval voting and in Section \ref{s:other} for several other voting systems.

\subsection{Common notations}

We start with the notation common to both the discrete space and continuous space frameworks; when possible, we use the notation from \cite{BCELP2016Handbook}.

The set of \emph{voter types} is assumed to be finite, of cardinal $n$, and denoted by $N$. It will be convenient to denote voter types by positive integers ($N=\{1,\dots, n\}$) in the theoretical part, but by upper-case letters from the end of the alphabet ($N=\{Z,Y,X,\dots\}$) in examples. Some of these letters will be used for maps below, but the context will prevent any confusion. In many examples below, voter types will coincide with preferences; however distinguishing gives more flexibility in the modelling, e.g. we can handle the case of different types of voters having the same preferences but using different heuristics.

Let $A$ be a finite set of \emph{candidates} (or \emph{alternatives}). Its cardinal is denoted by $m$ and assumed to be at least $2$ (most usually at least $3$). Candidates will be denoted by lower-case letters from the beginning of the alphabet ($A=\{a,b,c,\dots\}$). The set of weak orderings on $A$, i.e. rankings of candidates allowing ties, is denoted by $\mathcal{R}(A)$. The set of linear ordering on $A$, i.e. rankings of candidates without ties, is denoted by $\mathcal{L}(A)$.

\emph{Preferences} of each voter type will be modeled by an element of $\mathcal{R}(A)$; when disallowing ties, we will restrict to its subset $\mathcal{L}(A)$. A \emph{preference profile} is a pair $(P,w)$ where:
\begin{itemize}
\item $P=(\pref_1,\dots,\pref_n)\in\mathcal{R}(A)^n$ is a family of preferences, one for each voter type; the corresponding strict preferences are denoted by $\spref_i$,
\item $w=(w_1,\dots,w_n)\in \mathbb{R}_+^n$ is a family of \emph{weights} $w_i>0$, which represents the number of voters of each type. A usual situation is to take $w_i=1$ for all $i$ (or a positive integer, when grouping several voters in each type), but we allow more general weights for the continuous space framework.
\end{itemize}
A preference profile is represented by an array, each column of which represents a voter type, the column head carrying the weight (we shall often indicate the voter type on top for clarity) and the column listing the candidates in decreasing order of preference. With $3$ candidates $a,b,c$ an example of preference profile is:
\[
\begin{array}{ccc}
1 & 2 & 3 \\
100 & 101 & 102 \\
\hline
a & b & c \\
b & a & b \\
c & c & a
\end{array}
\]
meaning that $n=3$ and e.g. that $a\spref_1 b$ but $b\spref_2 a$, that there are 102  voters of type $3$, etc. For concision, an element of $\mathcal{L}(A)$ is denoted by listing the candidates in decreasing order of preference, e.g. $\pref:=acb$ means $a\spref c \spref b$. Elements of $\mathcal{R}(A)$ are denoted similarly with tied elements between parentheses.

Let $\mathcal{B}$ be a finite set of \emph{ballots} that can be cast by voters.
For example $\mathcal{B}$ can be $\mathcal{L}(A)$, or for Approval Voting the set $\mathcal{P}(A)$ of subsets of $A$.
An Approval Voting ballot $B\in \mathcal{P}(A)$ is said to be  \emph{sincere} for voters of type $i$ whenever
\[\forall a,b\in A \colon (b\in B \text{ and } a\spref_i b) \implies a\in B \]
i.e. when every candidate strictly preferred to any approved candidate is also approved.

Let $\mathcal{O}$ be a set representing the possible \emph{outcomes} of the election. We can let $\mathcal{O}=A$ if we are only interested in the winner, but it is convenient to take a larger set $\mathcal{O}$ to include all information that can be made available and used by voters to decide which ballot to cast. For example, for Approval Voting we can take $\mathcal{O}=[0,1]^A$, an element $r=(r_a, r_b,...)$ of $\mathcal{O}$ giving the proportion of voters approving of each candidate. We let $W:\mathcal{O}\to A$ be the function mapping a result to the corresponding winner (we shall always assume a tie-breaking rule, e.g. breaking ties in favor of the earlier candidate in the alphabetical order).

\subsection{Polling dynamics with discrete space}\label{s:DSPD}

We start with the usual case where in any given circumstance all voters of the same type cast the same ballot, which we call here the \emph{discrete space} framework.
A \emph{ballot profile} is thus an element of $\mathcal{B}^n$ and represents the ballots cast by each voter type.

We consider a \emph{social choice function}, also called \emph{voting rule} $f:\mathcal{B}^n\to A$ and assume the choice of $\mathcal{O}$ and $W$ makes it possible to decompose it as $f=W\circ g$ where $g:\mathcal{B}^n\to \mathcal{O}$ is called the \emph{information function}. Observe that by assuming a tie-breaking rule for $W$, we only consider \emph{resolute} social choice functions. For example, with the above choice of outcomes, the information function is given by
\[ g(B_1,\dots, B_n) = \bigg(\frac{1}{\sum_{i\in N} w_i}\sum_{i\in N \,|\, \alpha \in B_i} w_i\bigg)_{\alpha\in A}\]
and the social choice function maps $(B_1,\dots, B_n)$ to the candidates with the highest approval rating, ties broke by alphabetical order.

Each voter type $i$ is assumed to have a \emph{heuristic} $\sigma_i : \mathcal{O}\to \mathcal{B}$, representing the way voters of this type choose their next ballot given the information of the outcomes (which is itself determined from the previous ballot profile). Formally, the preferences of voter types do not appear in the Polling Dynamics to be defined in the next paragraph. Most usually, they do appear indirectly, through heuristics: more often than not, $\sigma_i(r)$ take the form $\sigma(\pref_i,r)$ where $\sigma$ is a map $\mathcal{R}(A)\times \mathcal{O}\to \mathcal{B}$ (or  $\mathcal{L}(A)\times \mathcal{O}\to \mathcal{B}$ if ties in preferences are not allowed). An example of such a map $\sigma$ is the Leader Rule of Laslier, which we mentioned earlier and recall below.

We obtain a \emph{Polling Dynamics}  on the set of ballot profiles
\begin{align*}
\varphi :\qquad\qquad \mathcal{B}^n &\to \mathcal{B}^n \\
(B_1,\dots,B_n) &\mapsto \big(\sigma_1(g(B_1,\dots,B_n)), \dots, \sigma_n(g(B_1,\dots,B_n)) \big)
\end{align*}
obtained by determining the outcome according to the ballot profile in argument and then applying each voter's heuristic to obtain a new ballot profile.

Alternatively, we can take as space the set of outcomes and consider the \emph{shifted Polling Dynamics}
\begin{align*}
\psi : \mathcal{O} &\to \mathcal{O} \\
 r &\mapsto g\big(\sigma_1(r), \dots, \sigma_n(r) \big)
\end{align*}
mapping an outcome to a new one, after applying each voter's heuristic then determining the resulting outcome.

These two maps are strongly related: $\psi\circ g = g\circ \varphi$ (in dynamical systems theory we say they are \emph{semi-conjugated}); $\varphi$ seems to carry more information, but their dynamical study are actually equivalent  because we assumed in the model that heuristics only depend on the outcome, not on the full details of the ballots.

By a \emph{dynamical equilibrium} we mean a fixed point of either $\psi$ of $\phi$, i.e.  either an outcome $r$ such that $\psi(r)=r$, or a ballot profile $(B_1,\dots, B_n)$ such that $\varphi(B_1,\dots, B_n)=(B_1,\dots, B_n)$. There is a one-to-one identification between these two points of view: when $r$ is fixed by $\psi$, then $(\sigma_1(r),\dots, \sigma_n(r))$ is fixed by $\varphi$ and sent back to $r$ by $g$; and when $(B_1,\dots, B_n)$ is fixed by $\varphi$, $g(B_1,\dots, B_n)$ is fixed by $\psi$ and sent back to $(B_1,\dots, B_n)$ by the heuristics.

\paragraph*{The Leader Rule and its modification.}
The Leader Rule is an example of heuristic for Approval Voting depending on tie-free preferences $\pref\in\mathcal{L}(A)$; assuming the expected winner $\omega(r)$ and the expected runner-up $\rho(r)$ can both be deduced from the outcome $r$, it is defined by
\[\mathrm{LR}(\pref,r) = \big\{\alpha\in A \,\big|\, \alpha\spref \omega(r) \text{ or }\alpha=\omega(r)\spref\rho(r)\big\};\]
in words, all candidates preferred to the expected winner are approved, and the expected winner is approved if and only if she is preferred to the expected runner-up.
When we say that voters (assumed to have tie-free preferences) apply the Leader Rule, we mean that their heuristics take the form $\sigma_i(r)=\mathrm{LR}(\pref_i,r)$.

When ties are allowed, i.e. $\pref{}\in\mathcal{R}(A)$, there are several possible ways to extend the Leader Rule. One could use the very same formula as above, but it would have the consequence that when the outcome is the ranking $abcd$ and the preferences are $(ab)(cd)$, the resulting ballot is blank. But in this situation, casting the ballot $\{a,b\}$ may prevent a worst candidate to be elected, and is always preferable to a blank ballot.
We will thus consider the following \emph{Modified Leader Rule} for preferences with ties by
\[\mathrm{MLR}(\pref,r) = \big\{\alpha\in A \,\big|\, \alpha\spref \omega(r) \text{ or } \alpha\simeq \omega(r) \spref \rho(r) \text{ or }(\alpha\simeq \omega(r) \text{ and }\nexists \beta\spref \alpha) \big\}  \]
in particular, compared to the above, when the expected winner is tied for top in the preferences, then she and all those tied with her are approved. Again, when we say that voters follow the Modified Leader Rule  we mean that $\sigma_i(r) = \mathrm{MLR}(\pref_i,r)$.

\subsection{Polling dynamics with continuous space}

We now turn to Polling Dynamics with continuous space (CS), the goal of which is to allow more flexibility in our assumption on voters' heuristics. In particular, we want to be able to consider a continuum of behaviors in each voter type, for example accounting for variable levels of bias, be it for example truth-bias (favoring sincere ballots) or lazy-bias (aversion to change one's ballot in view of new information). Example \ref{ex:robust_ex} illustrates this framework.

The discrete space of ballot profiles $\mathcal{B}^n$ thus has to be replaced by a more complicated object, encoding the proportion of each voter type casting each ballot. To this effect we use the \emph{simplex} $\Delta(X)$ over a finite set $X$, defined as
\[\Delta(X) = \Big\{(p_x)_{x\in X} \in [0,1]^X \,\Big|\, \sum_{x\in X} p_x = 1\Big\}.\]
The \emph{full continuous space} is then $\Delta(\mathcal{B})^n$; an element of this space is a doubly indexed family $(p^i_B)_{i\in N, B\in \mathcal{B}}$ where $p^i_B$ is the proportion of ballots $B$ cast among voters of type $i$, and is called a \emph{CS ballot profile}.

We now define a \emph{CS information function} as a map $G:\Delta(\mathcal{B})^n \to \mathcal{O}$, with the  \emph{associated CS social choice function} $F=W\circ G:\Delta(\mathcal{B})^n \to A$. For example, with the above choice of outcomes, for Approval Voting the CS information function is given by
\[ G\big((p^i_B)_{i,B} \big) = \bigg(\frac{1}{\sum_{i\in N} w_i}\sum_{\substack{i\in N\\ B\in\mathcal{B} \,|\, \alpha\in B}} p^i_B w_i\bigg)_{\alpha\in A} \]
where $p^i_B w_i$ represents the total number of voters of type $i$ casting the ballot $B$.

The full continuous space often has an unnecessarily large dimension, since under most heuristics voters will only cast ballots among a small subset $\mathcal{B}_i$ of $\mathcal{B}$ (one subset for each voter type;  for example, for Approval Voting we could restrict to ballots that are sincere with respect to the type's preferences). We thus define a more convenient \emph{continuous space} by
\[\mathcal{P} = \prod_{i\in N} \Delta(\mathcal{B}_i) \quad \subset  \Delta(\mathcal{B})^n\]
the elements of which are \emph{admissible CS ballot profiles}, or \emph{states}. We still denote by $G$ and $F$ the restrictions of the CS information function and social choice function to this space. Observe that the full CS appears as a special case by taking $\mathcal{B}_i=\mathcal{B}$ for all $i$.

A \emph{CS Polling Dynamics} is simply a map $\Phi : \mathcal{P} \to \mathcal{P}$. In general, it is not so much the dynamics of $\Phi$ that will be of interest, but the sequences of outcomes, which can be recovered as $G(\Phi^k(p_0))$ where $k$ is the number of iterations and $p_0=(p^i_B)_{i,B}\in \mathcal{P}$ is the initial ballots intended to be cast by the voters.

This very broad definition of a CS Polling Dynamics allows for much modeling flexibility, but the main interest is when such a map is deduced from some sort of heuristics of voters. We will not need this consideration in the robustness result which is very general, but let us give some definitions to show how this can be done.

We will assume that the choice of voters is based only on the predicted outcome and the ballots previously cast by voters of their type (this is enough to permit a fraction of voters not to adjust their choice, e.g. because they are not aware of the last poll). A \emph{CS heuristic} for voters of type $i$ is thus a map $\Sigma_i : \mathcal{O}\times\Delta(\mathcal{B}_i) \to \Delta(\mathcal{B}_i)$, and given a CS heuristic for each voter type the corresponding CS Polling Dynamics is given for all state $p=(p^i_B)_{i,B}\in \mathcal{P}$, setting $p^i=(p^i_B)_{B}$, by
\[\Phi(p) = \Big(\Sigma_1\big(G(p),p^1\big), \Sigma_2\big(G(p),p^2\big), \dots, \Sigma_n\big(G(p),p^n\big) \Big).\]

\section{Robustness of tie-free cycles}\label{s:robust}

In this section we want to prove that cycles of a discrete space Polling Dynamics that do not rely on the tie-breaking rule are robust, i.e. they persist under small enough perturbation. ``Small enough'' is a void concept in the discrete space setting, and we use CS Polling Dynamics to model the perturbations. The first step is to see that discrete space Polling Dynamics can always be realized as particular cases of CS Polling Dynamics.

\subsection{Embedding the discrete space Polling Dynamics into the continuous space.}

Observe that $\mathcal{B}^n$ embeds naturally into the full CS, by identifying $(B_1,\dots, B_n)$ with the element $\pi(B_1,\dots, B_n)=(p^i_B)_{i,B}\in \Delta(\mathcal{B})^n$ defined by
\[ p^i_B = \begin{cases} 1 & \text{when }B=B_i \\ 0 &\text{otherwise}  \end{cases} \]
Such an elements of the full CS, taking only the values $0$ and $1$, is called an \emph{extreme state}. We shall say that the CS information function $G$ \emph{extends} the information function $g$ whenever $G\circ\pi= g$, i.e. for all $(B_1,\dots,B_n)\in \mathcal{B}^n$, $G(\pi(B_1,\dots,B_n))=g(B_1,\dots,B_n)$. Without this assumption we would not be encoding the same voting method in both settings.

Consider fixed the (discrete space) heuristics $\sigma_i$ and the corresponding Polling Dynamics $\varphi$, and define the \emph{associated} CS Polling Dynamics by 
\[\Phi_0(p) = \pi(\sigma_1(G(p)),\dots,\sigma_n(G(p)),\]
i.e. starting from a CS ballot profile $p\in \Delta(\mathcal{B})^n$ we compute its outcome $G(p)$, then apply the heuristics to get new ballot choices for each voter type, then embed this into the full CS with the map $\pi$ to obtain a new CS ballot profile. Note that $\Phi_0$ is well-defined on the full CS, but all its values are extreme states.

The hypotheses we introduced in this section are precisely what needs to be assumed to ensure that the continuous space framework extends the discrete space situation: $\Phi_0$ produces the same sequences of outcome than $\varphi$ when starting at an extreme state.
\begin{lemm}
Assume that $G$ extends $g$ and that the heuristics only pick admissible ballots, i.e. $\sigma_i(r)\in\mathcal{B}_i$ for all voter type $i\in N$ and all outcome $r\in\mathcal{O}$. Then the associated CS Polling Dynamics takes its values in $\mathcal{P}$, and for all $(B_1,\dots, B_n)$ the sequences of outcomes
\[ g(\varphi^k(B_1,\dots, B_n)) \qquad\text{and}\qquad G(\Phi_0^k(\pi(B_1,\dots, B_n)))\]
are the same.
\end{lemm}

\begin{proof}
Since heuristics only pick admissible ballots, for all $p\in \Delta(\mathcal{B})^n$ we have 
\[\big(\sigma_1(G(p)),\dots,\sigma_n(G(p)\big)\in \prod_{i\in N} \mathcal{B}^i\]
so that $\Phi_0(p)\in \mathcal{P}$.

Since $G$ extends $g$, for all $(B_1,\dots,B_n)$ we have
\begin{align*}
\Phi_0(\pi(B_1,\dots, B_n)) &= \pi\big(\sigma_1(G(\pi(B_1,\dots, B_n))),\dots,\sigma_n(G(\pi(B_1,\dots, B_n))\big) \\
  &= \pi\big(\sigma_1(g(B_1,\dots, B_n)),\dots,\sigma_n(g(B_1,\dots, B_n)\big) \\
  &= \pi\big(\varphi(B_1,\dots, B_n)\big)
\end{align*}
By induction, we deduce that 
$\Phi_0^k\circ\pi=\pi\circ\varphi^k$
for all positive integer $k$. Using again that $G$ extends $g$, we get
$G\circ\Phi_0^k\circ\pi = G\circ \pi\circ \varphi^k = g\circ\varphi^k$, as desired.
\end{proof}

\subsection{Perturbations and robustness.}

The second step is to define what it means for two CS Polling Dynamics to be close one to the other.
We consider the metric induced on the full CS and on $\mathcal{P}$  by the supremum norm $\lVert\cdot\rVert$, i.e. given two CS ballot profiles $p=(p^i_B)_{i,B}$ and $\bar p=(\bar p^{\,i}_B)_{i,B}$ we set
\[ \lVert p-\bar p\rVert = \sup_{i\in N, B\in \mathcal{B}} \big\lvert p^i_B - \bar p^{\,i}_B \big\rvert.\]
This induces the usual topology on the full CS and, by restriction, on $\mathcal{P}$. We denote by $\bar B(p,\varepsilon)$ the closed ball of radius $\varepsilon\ge 0$ and center $p\in\mathcal{P}$ with respect to this metric.

We then consider the uniform metric, defined for every pair $\Phi_1,\Phi_2$ of CS Polling Dynamics by
\[D(\Phi_1,\Phi_2) := \sup_{p\in \mathcal{P}} \big\lVert \Phi_1(p)-\Phi_2(p) \big\rVert.\]

\begin{defi}
We say that a ballot profile $(B_1,\dots,B_n)$ is \emph{tie-free} (implicitly, with respect to heuristics $(\sigma_i)_i$, an information function $g$ and an extension $G$) whenever there exist a neighborhood $U$ of $\pi(B_1,\dots,B_n)$ in $\mathcal{P}$ such that for all $p\in U$:
\[W\circ G(p) = W\circ g(B_1,\dots,B_n) \quad\text{and}\quad \sigma_i(G(p)) = \sigma_i(g(B_1,\dots,B_n)) \quad\forall i\in N.\]

A ballot profile $(B^0_1,\dots,B^0_n)$ \emph{belongs to a tie-free cycle} for the Polling Dynamics $\varphi$ whenever $\varphi^k(B^0_1,\dots,B^0_n)=(B^0_1,\dots,B^0_n)$ for some positive integer $k$ (called a \emph{period} of the cycle) and the successive ballot profiles $(B^j_1,\dots,B^j_n)=\varphi^j(B^0_1,\dots,B^0_n)$ are all tie-free.
\end{defi}
Let us explain the rationale behind this definition. First, the condition $W\circ G(p) = W\circ g(B_1,\dots,B_n)$ means that the winner does not change if the considered CS ballot profile is close enough to $\pi(B_1,\dots,B_n)$. Second, all considered heuristics $\sigma:\mathcal{O}\to\mathcal{B}^n$ and their composition with the CS information function $\sigma\circ G:\mathcal{P}\to \mathcal{B}^n$ will be piecewise continuous, and thus piecewise constant (since they take value in a finite set). A state $p$ at which $\sigma\circ G$ is not continuous means that either a small change in the state can change the outcomes radically, or a small change in the outcome can change the resulting ballot profile for this heuristic. This correspond to breaking a tie, either in the information function or in the heuristic. For example in Approval Voting, when applying the Leader Rule with the above continuous choice of outcomes $\mathcal{O} = [0,1]^A$: when the first two candidates are tied, the tie-breaking rule embedded in $W$ is taken into account to determine the expected winner; and when several candidates are tied for runner-up, a tie-breaking rule must implicitly be written in the heuristic itself.

\begin{theomain}\label{t:robust}
Assume that $G$ extends $g$ and that the heuristics $\sigma_1,\dots,\sigma_n$ only pick admissible ballots. If the Polling Dynamics $\varphi$ has a tie-free cycle of period $k$, an element of which is denoted by $(B^0_1,\dots,B^0_n)$, then there exist $\varepsilon>0$ with the following property: for all CS Polling Dynamics $\Phi$ such that $D(\Phi,\Phi_0)\le \varepsilon$, for all $p^0\in\mathcal{P}$ such that $\lVert p^0-\pi(B^0_1,\dots,B^0_n)\rVert\le \varepsilon$, for all positive integer $j$,
\[W\circ G\big(\Phi^j(p^0)\big) = W\circ g\big(\varphi^j(B^0_1,\dots,B^0_n)\big).\]
If moreover $\Phi$ is continuous, then there exist $p^0\in\mathcal{P}$ such that $\lVert p^0 - \pi(B^0_1,\dots,B^0_n)\rVert\le \varepsilon$ and $\Phi^k(p^0)=p^0$.
\end{theomain}
In words, if the discrete space dynamics has a tie-free cycle, then perturbing the dynamics and starting point in the CS setting does not change the (periodic) sequence of winners; and if the perturbed CS Polling Dynamics is continuous, then this sequence of winners is furthermore realized by a cycle in $\mathcal{P}$, which is a perturbation of the  original cycle. We shall see that for some natural examples, the cycle moreover attracts a neighborhood of $p^0$ (see Examples \ref{ex:cycle} and \ref{ex:robust_ex}, which are easily generalized).

Note that Theorem \ref{t:robust} makes no rationality assumption: it applies not only to  best-response heuristics but allow arbitrary information-based heuristics, for example the pragmatist response policy of \cite{reijngoud2012response}.

\begin{proof}
For each positive integer $j$, set $B^j := (B^j_1,\dots,B^j_n) := \varphi^j(B^0_1,\dots,B^0_n)$. Since $B^j$ is tie-free, $\pi(B^j)$ admits a neighborhood $U^j$ (and we can choose the sequence $(U^j)_j$ to be $k$-periodic) such that for all $p\in U^j$,
\[W\circ G(p) = W\circ g(B^j) \quad\text{and}\quad \sigma_i(G(p)) = \sigma_i(g(B^j)) \quad\forall i\in N.\]
By definition of a neighborhood, for each $j$ there exist $\varepsilon^j>0$ such that $\bar B(\pi(B^j),\varepsilon^j)\subset U_j$; let $\varepsilon=\min(\varepsilon^0, \varepsilon^1,\dots, \varepsilon^{k-1})$.

The second tie-free condition $\sigma_i(G(p)) = \sigma_i(g(B^j))$ for all $i\in N$ and the definition of $\Phi_0$ ensures that for all $p\in U^j$, $\Phi_0(p)=\pi(B^{j+1})$.
Let $\Phi$ be a CS Polling Dynamics such that $D(\Phi,\Phi_0)\le \varepsilon$ and let $p^0\in\mathcal{P}$ such that $\lVert p^0-\pi(B^0_1,\dots,B^0_n)\rVert\le \varepsilon$. Then $\Phi_0(p^0)=\pi(B^1)$ and $\Phi(p^0) \in \bar B(\pi(B^1),\varepsilon)\subset U^1$. The same reasoning and an induction ensures that for all $j$, the state $p^j := \Phi^j(p^0)$ lies in $\bar B(\pi(B^j),\varepsilon)\subset U^j$. The first tie-free condition $W\circ G(p^j) = W\circ g(B^j)$ then gives the desired conclusion.

Assume further that $\Phi$ is continuous. Observe that $\mathcal{P}$ is a polyhedron of $\mathbb{R}^d$ for some $d$, and that balls $\bar B(s,\varepsilon)$ are thus convex sets, in particular homeomorphic to closed balls. For all $p\in \bar B(\pi(B^0),\varepsilon)$, we have $\Phi^{k-1}(p)\in U_{k-1}$ and thus $\Phi^k(p)\in \bar B(\pi(B^0),\varepsilon)$: $\Phi^k$ is a continuous map sending the topological ball $\bar B(\pi(B^0),\varepsilon)$ into itself. Brouwer's fixed point theorem ensures that there exist $p^0\in \bar B(\pi(B^0),\varepsilon)$ such that $\Phi^k(p^0)=p^0$.
\end{proof}

\begin{rema}
Theorem \ref{t:robust} applies to dynamical equilibria, by taking $k=1$.
\end{rema}

\begin{exem}\label{ex:cycle}
Consider Plurality voting with three candidates $A=\{a, b, c\}$, three voter types $N=\{X, Y, Z\}$ and preference profile
\[
\begin{array}{ccc}
X & Y & Z \\
4 & 2 & 3 \\
\hline
a & b & c \\
bc & ac & b \\
 &  & a
\end{array}
\]
since we use Plurality, $\mathcal{B}=A=\{a, b, c\}$; assume voters only consider the expected winner to choose how to vote, so that also $\mathcal{O}=A$ and $W$ is the identity map. The voting rule $f$ and the information function $g$ coincide; both take the triple of ballots $(v_X,v_Y,v_Z)\in\{a, b, c\}^3$ and return the candidate with the most votes.

For the discrete polling dynamics, we consider the following heuristics: $\sigma_X(\alpha)=a$ and $\sigma_Y(\alpha)=b$ for all $\alpha\in\{a, b, c\}$ (voters of type $X$ and $Y$ always vote for their favorite);
$\sigma_Z(a) = b$ and $\sigma_Z(\alpha)=c$ for $\alpha\in\{b,c\}$ (voters of type $Z$ vote for their favorite unless their least favorite threatens to win, in which case they settle for their second choice). The shifted polling dynamics is then given by
\[ \psi(a) = b,\quad \psi(b)=a, \quad \psi(c)=a \]
with a $2$-cycle $(a,b)$.
Let us now consider a natural perturbation of this dynamics. The full continuous space is $6$-dimensional: we need two numbers for each voting types, representing the proportions voting for two of the candidates (the third proportion being deduced from the first two). The CS information function $G$ compounds the votes in favor of each candidate, and returns the candidate with the most votes (tied broke by alphabetical order, say).

Assume $\mathcal{B}_X=\{a\}$, $\mathcal{B}_Y=\{b\}$ and $\mathcal{B}_Z=\{b, c\}$ (i.e. Voters of type $X$ and $Y$ still always vote for their favorite candidate, voters of type $Z$ can vote for either of their two preferred candidates). Then the continuous space $\mathcal{P}$ can be identified with $[0,1]$, a number $z\in [0,1]$ representing the proportion of voters of type $Z$ voting for $c$, and the restriction of $G$  to $\mathcal{P}$ is then given by
$G(z) = b$ when $z<\frac13$ and $G(z)=a$ when $z\ge \frac13$
(indeed $a$ always receives $4$ votes, while $c$ receives $3z\le 3$ and $b$ receives $5-3z$). We consider a parameter $t\in[0,1]$ representing the reluctance of voters to change their votes and define a corresponding CS heuristic by 
\[\Sigma_Z^t(\alpha,z) = \begin{cases} 1-t(1-z) &\text{when }\alpha\neq a \\ tz &\text{when }\alpha=a \end{cases}\]
(a proportion $t$ of voters of type $Z$ who should switch votes under the discrete heuristic $\sigma_Z$ do not do so under $\Sigma_Z^t$). We denote by $\varphi$ the discrete Polling Dynamics and by $\Phi_t$ the CS Polling dynamics driven by $\Sigma_Z^t$. The CS Polling Dynamics $\Phi_0$ associated to $\varphi$ is obtained as follows: given $z\in[0,1]$, we compute the expected winner $G(z)$ and determine the behavior of voters with the discrete heuristics: when $z<\frac13$, all voters of type $Z$ shall vote for $c$ in the next round so that $\Phi_0(z)=1$ while when $z\ge\frac13$, they all vote for $b$ so that $\Phi_0(z)=0$. Every starting point $z$ is immediately mapped to an extreme state, as always, and in this particular case they all are attracted to the cycle $(a,b)$. For $t\in[0,1]$, we get $\Phi_t(z) = 1-t+tz$ when $z<\frac13$ and $\Phi_t(z)=tz$ when $z\ge\frac13$ (in particular the $t=0$ matches $\Phi_0$); the graph of $\Phi_t$ consist in two lines of slope $t$, as pictured in Figure \ref{f:simple_example}. The map $t\mapsto \Phi_t$ is continuous in the uniform topology, so that Theorem \ref{t:robust} ensures that for $t$ small enough and $z_0$ close enough to $0$ or $1$, the iteration of $\Phi_t$ starting at $z_0$ alternates between the intervals $[0,\frac13)$ and $[\frac13,1]$, producing the same cycle of outcomes than the discrete Polling Dynamics.

Here we can compute exactly how large a perturbation we can afford. For all $t\in[0,\frac12)$, there is a $2$-cycle consisting of the points $\frac{t}{1+t}<\frac13$ and $\frac{1}{1+t} >\frac23$. Since the slope of the graph of $\Phi_t$ is less than one, this $2$-cycle attracts (exponentially fast) all nearby orbits. At $t=\frac12$, one can check that orbits converge to an almost-cycle $(\frac13,\frac23)$; this is not a true cycle since $\frac13$ is sent to $\frac16$ because of the tie-breaking rule, but points $z<\frac13$ close to $\frac13$ are sent near $\frac23$. One can check further that for $t>\frac12$, 
no orbit realizes the sequence $ababa\dots$ of winners (but $aaa\dots$ and $bbb\dots$ are not realized either, the winner still alternate between $a$ and $b$, but not as regularly as for $t<\frac12$).

\begin{figure}[htp]
\begin{center}
\includegraphics[scale=1]{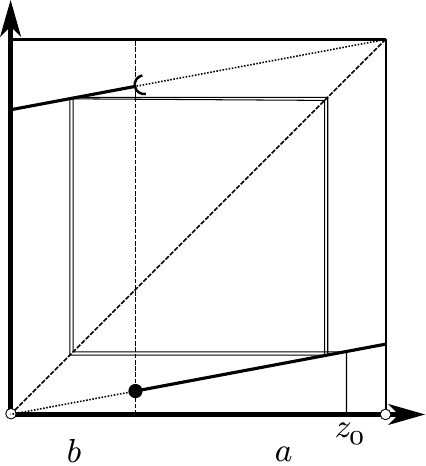}
\caption{Graph of $\Phi_t$ and an orbit. Dotted: construction lines. Dashed: diagonal and delimitation between the two possible outcomes.}\label{f:simple_example}
\end{center}
\end{figure}
\end{exem}

\begin{rema}
As mentioned by one of the anonymous referees, it seems Theorem \ref{t:robust} should apply to asynchronous iterative voting. It would need to change the definitions and notation to adjust to this case; let us briefly describe how one could proceed.

The usual discrete model for asynchronous iterative voting yields an oriented graph, with an arrow from one ballot profile to another when an admissible move takes the former to the latter. This can be modeled as a set-valued function on the set of ballot profiles. The continuous space dynamics could thus be any set-valued function on the CS space we defined here. The discrete model would  embed in that CS space (and give a particular example of CS dynamics) as follows: given a state, first we apply the CS information function $G$, and then we look at all admissible moves given this result. We can either assume all voters of a given type apply the chosen move, or only a (non-zero) ``quantum'' of them (rounding if necessary in such a way that every paths ultimately reflects a paths occurring in the original discrete model). In both cases we get a set-valued map on the CS space. Perturbations can be measured by the uniform Hausdorff distance, i.e. when $\Phi$ and $\Phi'$ are set-valued maps, $D(\Phi,\Phi')$ would be the least $\varepsilon$ such that for all state $p$, every element of $\Phi(p)$ is at distance at most $\varepsilon$ from some element of $\Phi'(p)$, and every element of $\Phi'(p)$ is at distance at most $\varepsilon$ from some element of $\Phi(p)$.

Then a variant of Theorem \ref{t:robust} can be expected, with basically the same proof: for each cycle in the discrete model that stays away from the frontiers between candidates' winning regions, if we consider a small enough perturbation we shall find a nearby CS cycle with the same sequence of winners. \end{rema}

\section{Bad cycles in Approval Voting}\label{s:AVexamples}

Let us recall some definitions relative to what can be considered ``good'' or ``bad'' collective choices among the candidates. 

Given a preference profile, a candidate $\alpha$ is said to \emph{dominate} a candidate $\beta$ (sometimes written $\alpha >\beta$, but beware that this is not a transitive relation) whenever there are strictly more voters that strictly prefer $\alpha$ to $\beta$ than voters that strictly prefer $\beta$ to $\alpha$ (i.e. in the \emph{majority graph}, assuming indifferent voters abstain, there is an arrow from $\alpha$ to $\beta$).
A candidate $\alpha$ is then said to be a \emph{Condorcet winner} whenever she dominates every other candidate; a Condorcet winner may or may not exist, but if she exists she is unique. When preferences have no ties, this is the usual definition of a (strong) Condorcet winner; when there are ties, a stronger definition could be possible: to dominate, one could ask for a majority of all voters, including abstainers. Similarly a candidate $\beta$ is said to be a \emph{Condorcet loser} whenever she is dominated by every other candidate; again, a Condorcet loser may or may not exist and is unique if she exist.

Last, we will use a stronger notion than Condorcet loser: a candidate is said to be an \emph{absolute majority loser} whenever she is a Condorcet loser and there is a strict majority of the electorate that ranks her last (possibly tied with others) in their preferences. An absolute majority loser rarely exists.

\subsection{First example: non-convergence of the Leader Rule}
\label{s:example1}

In this Section all heuristics will be given by the Leader Rule, which only relies on the identity of the expected winner and runner-up, not on the expected scores of candidates (its definition is recalled in Section \ref{s:DSPD}). We can thus consider a simple set of outcomes, $\mathcal{O}=\{\omega\rho\colon \omega\neq \rho \in A\}$ the set of order pairs of distinct candidates, and the winner map is defined by $W(\omega\rho)=\omega$; $\rho$ represents the runner-up. Preferences of voter types are without tie (i.e. in $\mathcal{L}(A)$) as needed for the Leader Rule. The information function $g$ is obtained by ranking candidates by decreasing order of approval numbers ($r_\alpha = \sum_{i \mid \alpha\in B_i} w_i$), breaking ties by alphabetical order, and then selecting the two first candidates, keeping their order. The CS information function $G$ is defined similarly, with approval numbers 
\[r_\alpha = \sum_{\substack{i\in N \\ B \mid \alpha\in B}} p^i_B w_i,\] making it an extension of $g$. The admissible ballots of type $i$ are the sincere ballots with respect to the preference order $\pref_i$; the LR ensures that heuristics only pick admissible ballots.
Our goal is to prove the following result.
\begin{theomain}\label{t:LR}
Using Approval Voting, there exists a preference profile for $4$ candidates such that:
\begin{itemize}
\item there is a Condorcet winner,
\item each voter has preferences without ties,
\item assuming voters follow the Leader Rule, the Polling Dynamics has a cycle along which the Condorcet winner is never elected. Moreover, a majority of the initial ballot profiles lead to this cycle, among those who are both sincere and expressive (i.e. approving a non-empty, strict subset of candidates).
\end{itemize}
\end{theomain}

\begin{proof}
We set $A=\{a,b,c,d\}$ and consider the following preference profile with $7$ types of voters:
\[\begin{array}{ccccccc}
T   & U   & V   & W   & X   & Y   & Z   \\
100 &1000 &1001 &1002 &1004 &1008 &1016 \\
\hline
a   & b   & c   & d   & b   & c   & d   \\
b   & a   & a   & a   & c   & d   & b   \\
c   & c   & d   & b   & a   & a   & a   \\
d   & d   & b   & c   & d   & b   & c 
\end{array}\]
and heuristics given by the Leader Rule.

The voters of type $U,V,W$ like $a$ but each prefers one of $b,c,d$ better, while the voters of type $X,Y,Z$ do not like $a$ too much but distaste one of $b,c,d$ even more, creating a cycle in the majority graph $b \spref c \spref d \spref b$  with $a$ close to tie with each of $b,c,d$. Meanwhile, voters of type $T$ prefers $a$ to any other candidate, and their moderate number suffice to make $a$ a Condorcet winner while maintaining the cycle $b,c,d$ in the majority graph. The precise numbers of types $U$ to $Z$ are chosen, for the sake of fanciness, to exclude any perfect tie (different sums of distinct powers of $2$ never agree).

Figure \ref{f:diagram} represents the shifted Polling Dynamics $\psi$ in the form of a graph. We only give the details of the computations along the cycle, others are similar.
\begin{figure}[htbp]
\begin{center}
\includegraphics[width=\linewidth]{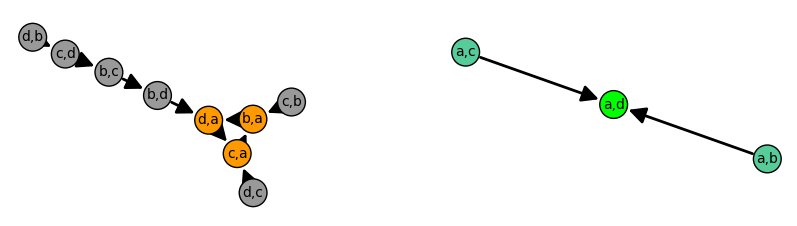}
\caption{The shifted Polling Dynamics of the first example. Outcomes where the Condorcet winner is elected are shown in green, light green for the dynamical equilibrium. The periodic outcomes not electing the Condorcet winner are shown in orange.}\label{f:diagram}
\end{center}
\end{figure}

Consider the outcome $ba$. Under the Leader Rule, it leads to the following ballot profile and results:
\begin{align*}
B_T &= \{a\} &  B_U &= \{b\}   &  B_X &= \{b\}  & r_a &=3111 & r_b &=3020\\
&& B_V &= \{c,a,d\}  & B_Y &= \{c,d,a\}                &   &      & r_c &=2009\\
&& B_W &= \{d,a\}  & B_Z &= \{d,b\}                &   &      & r_d &=4027
\end{align*}
so that $\psi(ba)=da$ -- i.e. $a$ stays second, while the previously unthreatening $d$ comes in first position. The strategic adjustments triggered by the outcome $da$ are as follows:
\begin{align*}
B_T &= \{a,b,c\}  &  B_U &= \{b,a,c\}  & B_X &= \{b,c,a\}  & r_a &=3105 & r_b &= 2104\\
&&             B_V &= \{c,a\}  & B_Y &= \{c,d\}  &   &      & r_c &= 4113\\
&&             B_W &= \{d\}  & B_Z &= \{d\}  &   &      & r_d &= 3026
\end{align*}
so that $\psi^2(ba) = \psi(da) = ca$. The corresponding strategic adjustments are then:
\begin{align*}
B_T &= \{ab\}  &  B_U &= \{b,a\}  & B_X &= \{b,c\}  & r_a &:3118 & r_b &: 4122 \\
&&             B_V &= \{c\} & B_Y &= \{c\}  &   &      & r_c &: 3013\\
&&             B_W &= \{d,a,b\}  & B_Z &= \{d,b,a\}  &   &      & r_d &: 2018
\end{align*}
so that $\psi^3(ba)=ba$.

The graph shows that more than two third of the outcomes lead to the bad cycle. A computer assisted enumeration of all possible sincere, expressive ballot profiles shows that $1353$ out of $2187$, i.e. more than $60\,\%$, lead to the bad cycle.
\end{proof}

\begin{rema}
It is worthwhile to observe (although a bit tedious to check) that in this example, \emph{all} outcomes can be obtained from a certain sincere ballot profile. For example, $dc$ is obtained if all voters cast the ballot with only their preferred candidates. That all outcomes are covered is hardly surprising: there are only $12$ of them, but $3^7=2187$ combinations of sincere ballots (excluding abstentions).
\end{rema}

\begin{rema}
In the cases of $m>4$ candidates, we can take the above example and add $n-4$ dummy candidates that appear at the end of all voters preferences. The only property that may not be preserved in this operation is the size of the basin of attraction of the $3$-cycle: for example the outcomes where one of the dummy candidates wins will all be sent by $\varphi$ to an outcome where $a$ wins, since voters would vote for all of $a, b, c, d$. However this is easily fixed by adding a voter type, for example in the case of a fifth candidate $e$ one could take $50$ voters of type $S$ with preference order $becda$. Indeed, this voters will break the tie between $a, b, c, d$ whenever $e$ is the expected winner, in favor of $b$, thus leading to the basin of attraction of the $3$-cycle.
\end{rema}

\begin{rema}
As pointed out by one of the anonymous referees, Theorem \ref{t:LR} does not hold for $3$ candidates. Assume there is a Condorcet Winner $a$ and two other candidates $b,c$. Since preferences have no ties, there are $6$ possible preferences, and we can reduce to six voter types; let $T_{\alpha\beta}$ denotes the voters ranking $\alpha$ first and $\beta$ second, and let it also denotes the number of such voters. We assume there is no equality in any partition into two groups of three voter types, which is a generic condition, to avoid dealing with ties. Since $a$ is a Condorcet winner, 
\begin{align*}
T_{ab} + T_{ac} + T_{ca} &> \frac12 > T_{bc} + T_{ba} + T_{cb} \\
T_{ab} + T_{ac} + T_{ba} &> \frac12 > T_{bc} + T_{ca} + T_{cb}
\end{align*}
One of $b$ and $c$ would win a duel against the other, and without loss of generality we assume $b$ does, i.e.
\[T_{ba} + T_{bc} + T_{ab} > \frac12 > T_{ca} + T_{cb} + T_{ac}.\]
We assume that all voters follow the Leader Rule and prove that the Polling Dynamic converges to one of the outcomes $ab$ or $ac$.

First, whatever the expected outcome is, $a$ receives either $T_{ab} + T_{ac} + T_{ca}$  or $T_{ab} + T_{ac} + T_{ba}$ votes, depending on whether $b$ is above $c$ or not, in any case more than half.

When $a$ is expected winner, $b$ and $c$ receives $T_{bc} + T_{ba} + T_{cb}<\frac12$ and  $T_{bc} + T_{ca} + T_{cb}<\frac12$ votes respectively, so that $a$ stays ahead and either $ab$ or $ac$ is an equilibrium, attracting the other one of these two outcomes.

Whenever $a$ is second, none of the voters for $a$ vote for the expected winner, and it follows that in the next round, the contender will be behind $a$. The third candidate will receive the votes of voters preferring her to the expected winner; in the case of $c$, it means either $T_{ca}+T_{cb}+T_{ac}$ or $T_{ca}+T_{cb}+T_{bc}$ votes, according to whether $b$ or $a$ is expected to win, and both are less than $\frac12$. It remains to consider the outcome $cab$, where $b$ receives the votes $T_{ba}+T_{bc}+T_{ab}$. In this case, we can have a sequence of outcomes $ca,ba$ but then $a$ wins in the next outcome. Whenever $a$ gets at least second, she will therefore become first and stay first ever after.

Now, in every case where $a$ is third, she will receive more than half the votes and $c$ will receive less than half, so that $a$ will get at least the second position in the next round, after which she will take the lead. This ends the proof of convergence to equilibrium.
\end{rema}

\begin{rema}
If we do not ask for a Condorcet winner to exist, a cycle can be produced with $m=3$ candidates, e.g. with the preference profile
\[\begin{array}{ccc}
X  & Y  & Z  \\
10 & 11 & 12 \\
\hline
a  & b  & c  \\
b  & c  & a  \\
c  & a  & b 
\end{array}\]
where starting with voters approving only their top candidate, we obtain a cycle of outcomes $cb$, $bc$, $ab$.
\end{rema}

From Theorem \ref{t:robust} we deduce that the cycle persists when we perturb the Polling Dynamics in the CS setting (to be meaningful, here we change back the set of outcomes to $\mathcal{O}=[0,1]^A$).
\begin{coro}
Consider a preference profile as given by Theorem \ref{t:LR}. Let $\Phi_0$ be the induced CS Polling Dynamics on the continuous space $\mathcal{P}$ defined by making admissible exactly the sincere ballots. 
Then there exist $\varepsilon>0$  and an open subset $U$ of $\mathcal{P}$ such that for \emph{all} CS Polling Dynamics $\Phi$ with $D(\Phi,\Phi_0)\le\varepsilon$ and all $p^0\in U$, the sequence of winners
\[W\circ G(\Phi^j(p^0)) \qquad j=0,1,2,\dots\]
is periodic and does not contain the Condorcet winner.
\end{coro}

\subsection{Second example: the possible election of an absolute majority loser}\label{s:example2}

Our second example, at the small cost of introducing a modification of the LR accounting for ties, improves on the previous one on two accounts: it necessitates only $3$ candidates, and it exhibits a cycle where an absolute majority loser could get elected.
We conserve most the setting of the previous Section: $\mathcal{O}$, $W$, $g$ and $G$ are as above. The only change is that we allow for ties in preferences and more varied heuristics.

\begin{theomain}\label{t:example2}
Using Approval Voting, there exists a preference profile on $3$ candidates with ties in preferences allowed and sincere heuristics such that:
\begin{itemize}
\item there are a Condorcet winner and an absolute majority loser,
\item the Polling Dynamics has a $2$-cycle, one of whose ballot profiles elects the absolute majority loser,
\item there is an equilibrium not electing the Condorcet winner,
\item moreover only one of four sincere, expressive ballot profiles avoid the above bad cycle and equilibrium.
\end{itemize}
\end{theomain}

\begin{proof}
We consider the following preference profile:
\[\begin{array}{cccc}
Z   & Y   & X   & W   \\
3 & 1   & 3 & 5 \\
\hline
a   & a   & b   & c   \\
b   & bc & a   & ab\\
c   &     & c   &   
\end{array}\]
with each voter type $i\in N$ using the Modified Leader Rule. In particular voters of type $W$ will not choose between $a$ and $b$, thus always casting the ballot $\{c\}$, no matter which outcome is expected. Similarly, voters of type $Y$ always cast the ballot $\{a\}$ (this last type is only introduced here for tie-breaking).

Note that $a$ is a Condorcet winner, beating $b$ with a score of $4$ to $3$ (voters of type $W$ abstaining) and $c$ with a score of $7$ to $5$. Moreover $c$ is a Condorcet loser, loosing to $b$ by $6$ to $5$; actually $c$ is a worst candidate of $7$ out of the $12$ voters, making it an absolute majority loser.

Assume as starting expected outcome the result obtained if each voter votes for every candidates she does not rank last:
\begin{align*}
B_Z   &= \{a,b\}&  B_Y &= \{a\}  &  B_X   &= \{b,a\}      &  B_W &= \{c\} & \\
&&&&&&&& r_a &= 7  &  r_b &= 6  &  r_c &= 5
\end{align*}
leading to $a$ being expected winner and $b$ expected runner-up (corresponding to the Condorcet order). This leads voters of type $Z$ and $X$ to adjust their votes: their favorite candidate is either threatened by their second-favorite (for $Z$) or have a shot at winning the election from a current runner-up position (for $X$). Consistently with their heuristics they choose to vote only for their favorite candidate:
\begin{align*}
B_Z   &= \{ a \}  &  B_Y &= \{a\}  &  B_X   &= \{b\}    &  B_W &= \{c\} & \\
&&&&&&&& r_a &= 4  &  r_b &= 3  &  r_c &= 5.
\end{align*}
The second poll thus results in a win of $c$ with $a$ as runner-up. This result induces voters of type $Z$ and $X$ to resume approving both $a$ and $b$, in order not to let $c$ be elected (again, this is a consistent application of the modified Leader Rule). This results in the same ballots being cast as in the first poll, so we get a $2$-cycle, in which the worst candidate is elected in one of the outcomes.

Figure \ref{f:diagram2} represents the shifted Polling Dynamics $\psi$ in graph form.
\begin{figure}[htbp]
\begin{center}
\includegraphics[width=\linewidth]{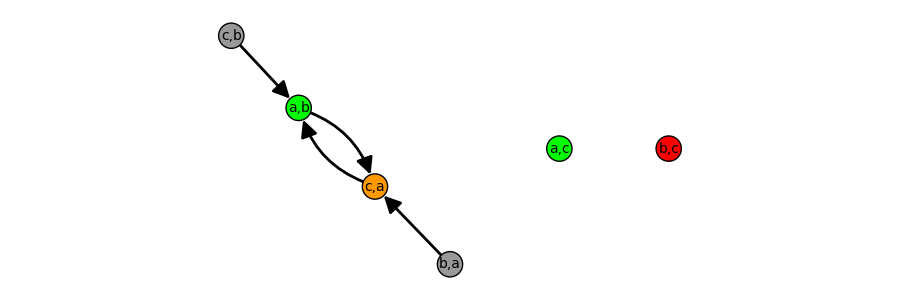}
\caption{The shifted Polling Dynamics of the second example. Light green outcomes are accessible after an arbitrary large number of iterations and elect the Condorcet winner; orange (periodic) and red (fixed point) outcomes are accessible after an arbitrarily large number of iterations but do not elect the Condorcet winner. The orange outcome is arguably the most problematic, as it elects the absolute majority loser $c$.}\label{f:diagram2}
\end{center}
\end{figure}

We see that of $6$ outcomes, $4$ lead to the cycle that can elect either the Condorcet winner $a$ or the absolute majority loser $c$ depending on whether the number of polls conducted before the election is odd or even, while the other $2$ are equilibria, one electing the Condorcet winner $a$ and the other the Condorcet runner-up $b$. There are $4$ sincere, expressive ballot profiles (two possible choices for each of $Z$ and $X$), one for each of the equilibria and each of the $2$-cycle outcomes.
\end{proof}

\begin{rema}\label{r:strongCW}
If we accept to depart further from the LR, we could avoid ties in preferences and preserve the features of the examples by splitting $W$ into two types of voters of equal size, with respective preferences $cab$ and $cba$ and heuristic to always vote $\{c\}$ (with the interpretation that these voters prefer $c$ to the other two by far, but still have a slight preference between $a$ and $b$). This shows that Theorem \ref{t:example2} also holds with the stronger definition of Condorcet winner (see beginning of Section \ref{s:AVexamples}).
\end{rema}


\begin{rema}
An argument that could be raised against this example is that it needs that a large proportion of voters having $c$ as favorite candidate would never vote for any other candidate. While this is indeed a crucial feature of the voters preferences in this example, there are two counter-arguments. First, this situation seems not all that unlikely: far-right candidates with a strong anti-establishment discourse can have many supporters who would consider all other candidates (or at least those with a chance of being elected) as part of the very same ``establishment'' and thus would only approve of $c$. Second, this can be a textbook case of manipulation by a coalition: if the minority of all voters who prefer $c$ (with preferences $cb$ or $ca$ say) gather in a coalition and decide to vote only for $c$, they get a good chance to have $c$ elected against the will of a two-third majority! Actually, these counter-arguments feed on each other: an anti-establishment discourse can serve the purpose of forming a coalition-in-practice of voters who will not express their preferences between $a$ and $b$ in order to favor $c$.
\end{rema}

Again, Theorem \ref{t:robust} ensures that this bad cycle is robust. Let us make this more concrete by considering an explicit perturbation.

\begin{exem}\label{ex:robust_ex}
We now give a CS setting extending the discrete-space example in the proof of Theorem \ref{t:example2}. Let the sets of admissible ballots be 
\[\mathcal{B}_Z=\big\{\{a\},\{a,b\}\big\}, \quad \mathcal{B}_Y=\big\{\{a\}\big\}, \quad \mathcal{B}_X=\big\{\{b\},\{a,b\}\big\}\quad \text{and}\quad\mathcal{B}_W=\big\{\{c\}\big\}\]
and consider outcomes giving the proportion of votes received by each candidates: $\mathcal{O}=[0,1]^A$.

Observe that since the simplex over a singleton is a singleton, and the simplex over a two-element set is an interval, $\mathcal{P}$ can be identified with the square $[0,1]^2$, with coordinates $(x,z)$ where $x$ (respectively $z$) represents the proportion of voters of type $X$ (respectively $Z$) casting the ballot $\{a,b\}$.
The leftmost part of Figure \ref{f:continuous} represents $\mathcal{P}$. The three lines corresponding to ties (of equation $(z=x+\frac13)$ for $a$ and $b$; $(x=\frac13)$ for $a$ and $c$; $(z=\frac23)$ for $b$ and $c$); they are concurrent at the point where all three candidates are tied, and delimit six areas, one for each possible outcome (all of which are possible under the chosen preference profile and restriction of ballots, as we can see in the figure). We denote by $A_{abc}$ the region where the outcome is $abc$ (boundary segments are attributed according to the tie-breaking rule), and similarly for the other five regions.

\begin{figure}[htp]
\begin{center}
\includegraphics[scale=1]{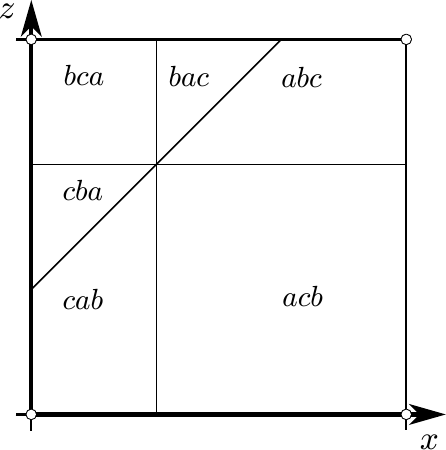}
\,\includegraphics[scale=1]{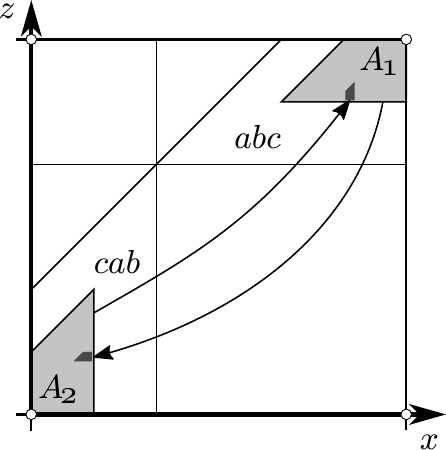}
\caption{A continuous-space example. \emph{Left:} the state space $\mathcal{P}$. The corners are the extreme states, corresponding to the four outcomes that are attainable from a ballot profile under the discrete-space Polling Dynamics.
\emph{Right:} on the same state space, any CS Polling Dynamics where $85\%$ of voters adjust their ballot according to the modified Leader Rule when candidates are separated by $4\%$ margins will send the light-grey regions $A_1$ and $A_2$ into each other (images shown in dark grey), thus ensuring a $2$-cycle with one state near the upper-rigth corner, one near the lower-left corner.}\label{f:continuous}
\end{center}
\end{figure}

Given $\delta\in [0,1]$, we let $\mathcal{T}_\delta$ be the set of $(x,z)\in\mathcal{P}$ such that $(r_a,r_b,r_c):=V(x,z)$ avoids ties by a margin at least $\delta$, i.e. $\lvert r_\alpha-r_\beta\rvert \ge \delta$ for all $\alpha\neq\beta\in A$. We consider any perturbation $(\Phi_\varepsilon)_{\varepsilon\in [0,1]}$ of the embedded discrete space Polling Dynamics such that:
\[ \Phi_\varepsilon(x,z) = (1-\varepsilon)\Phi_0(x,z) + \varepsilon (x,z) \qquad \forall (x,z)\in\mathcal{T}_\delta\]
which models the situation where after each poll where no two candidates are $\delta$-close to be tied, a proportion $\varepsilon$ of voters keep their ballots unchanged and the remaining $(1-\varepsilon)$ apply the heuristics of their types defined in the discrete setting (here the LR for types $X$ and $Z$, types $Y$ and $W$ always voting only for their top candidate). 

We will show that taking $\delta=.04$ and $\varepsilon = .15$ is small enough for the conclusion of Theorem \ref{t:robust} to apply. Figure \ref{f:continuous-numeric} shows the first five iterations of such a perturbation for $\delta=.04$ and $\varepsilon=.15$, and we can observe a bad $2$-cycle with points in $A_{abc}$ and $A_{cab}$.

\begin{figure}[htp]\begin{center}
\includegraphics[width=.32\linewidth]{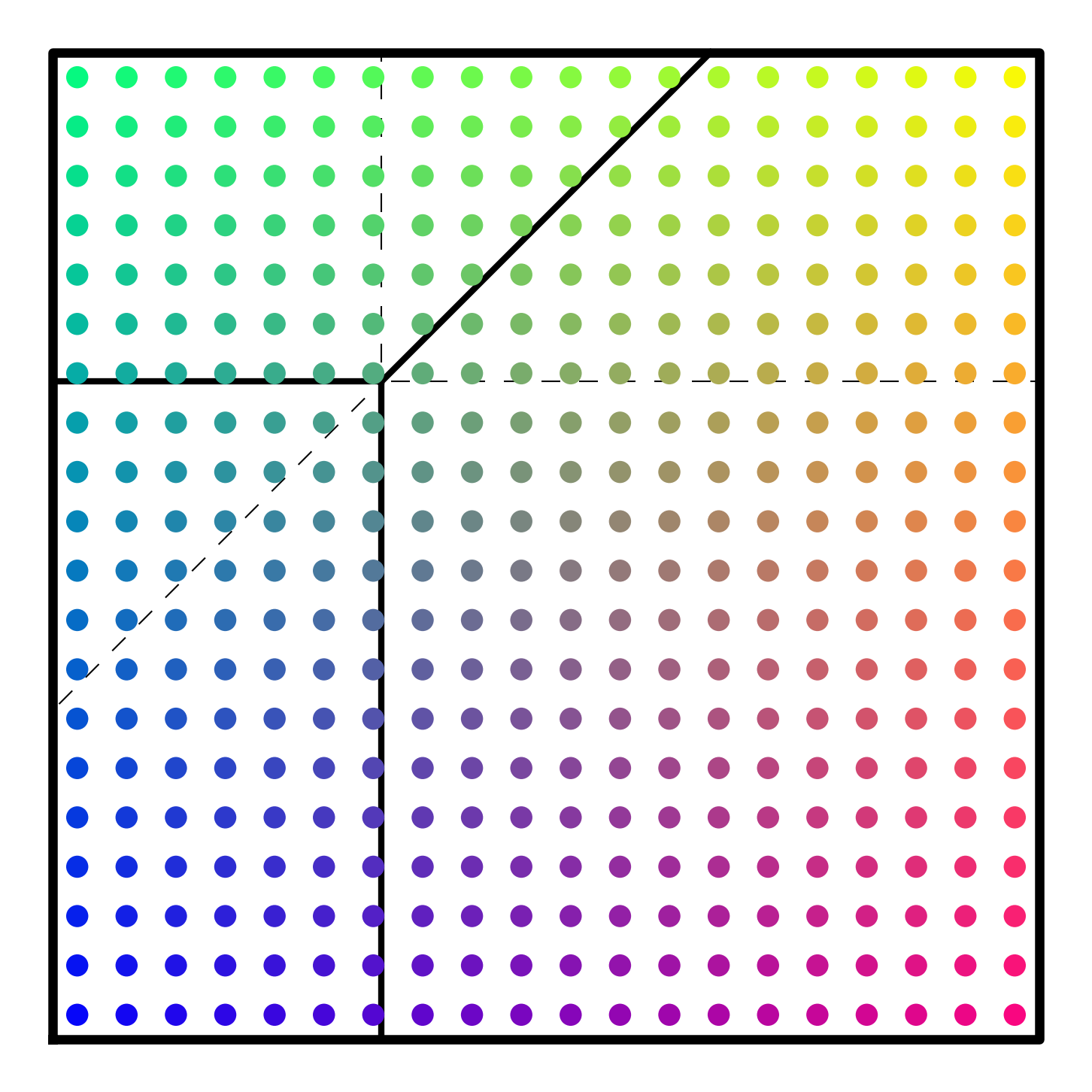}
\includegraphics[width=.32\linewidth]{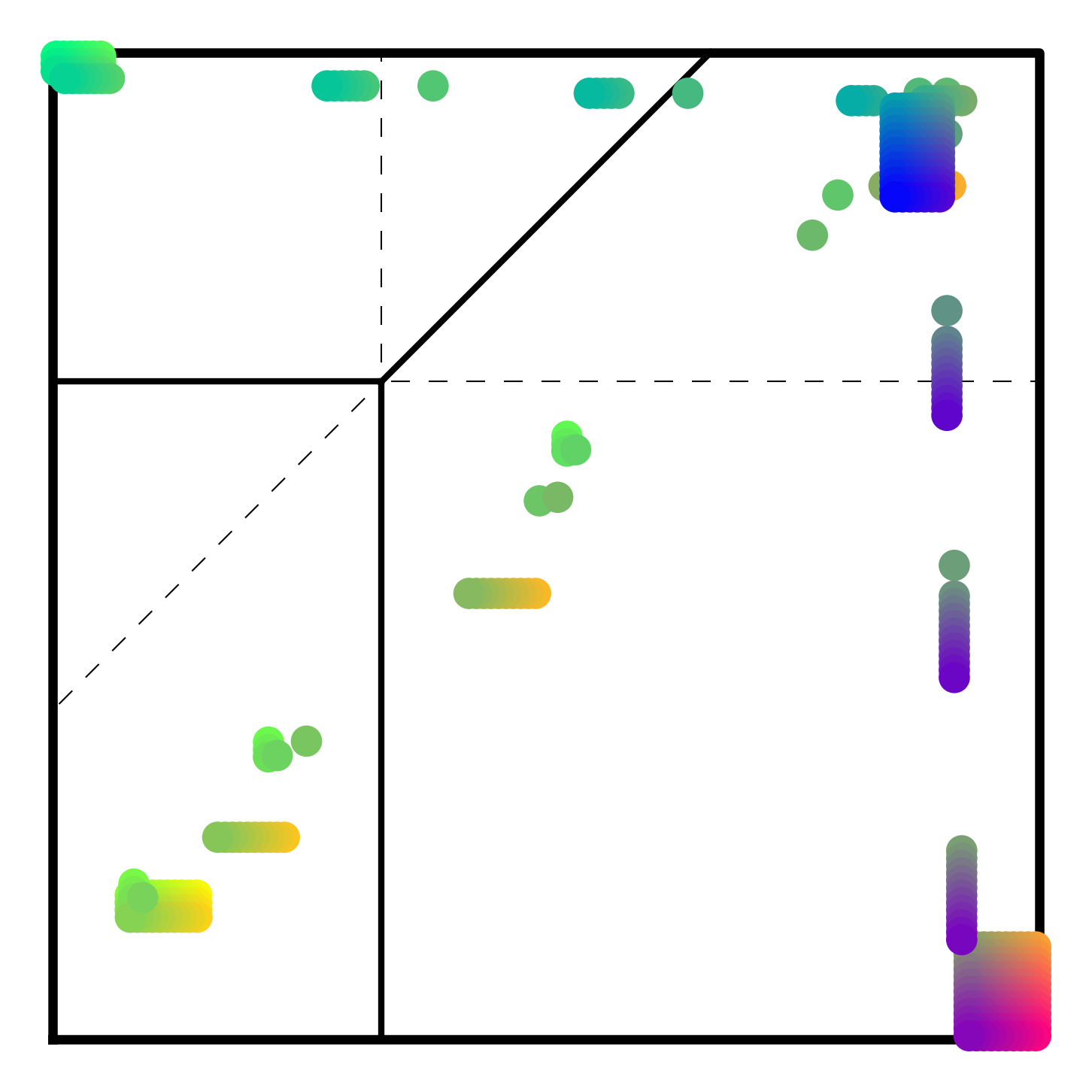}
\includegraphics[width=.32\linewidth]{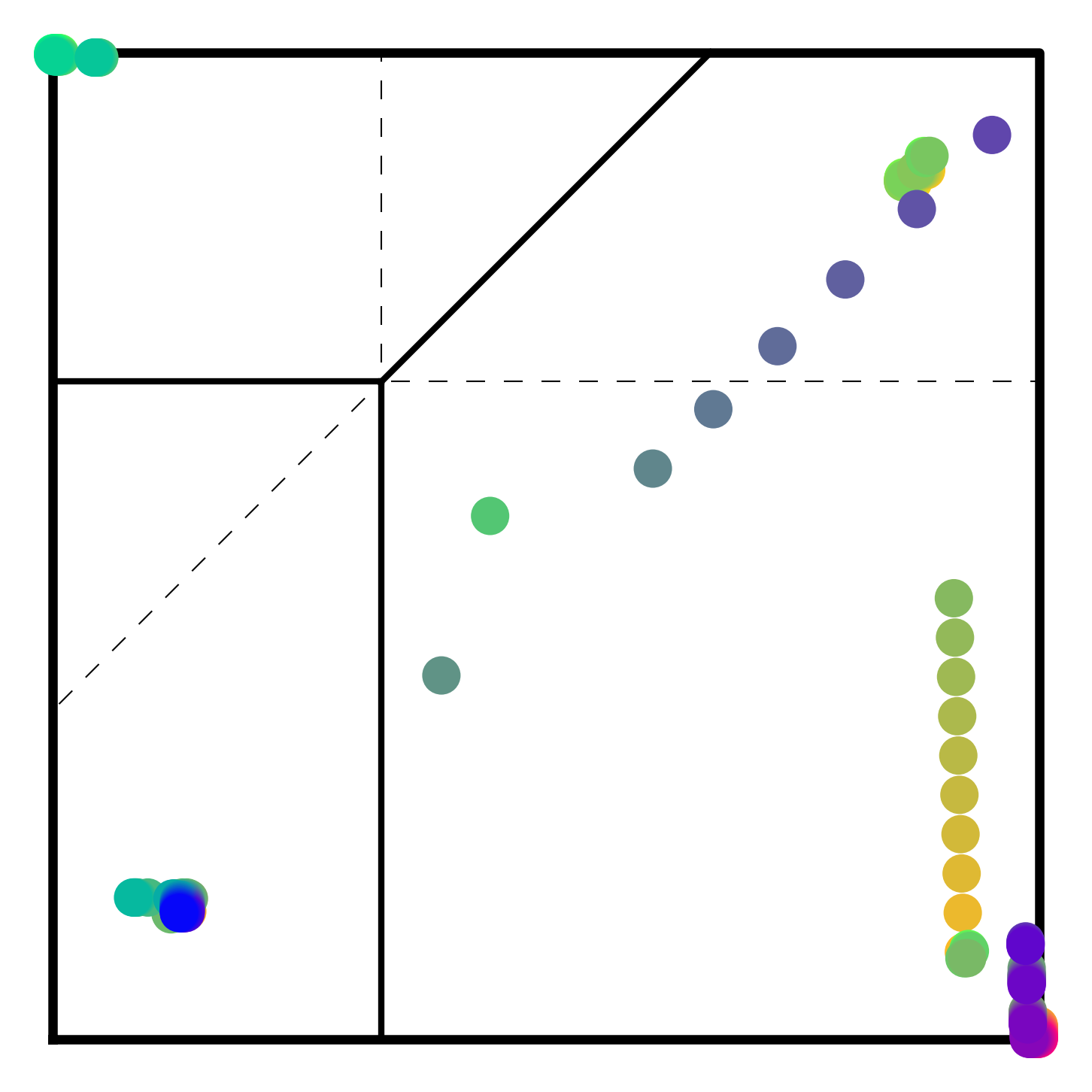}\\
\includegraphics[width=.32\linewidth]{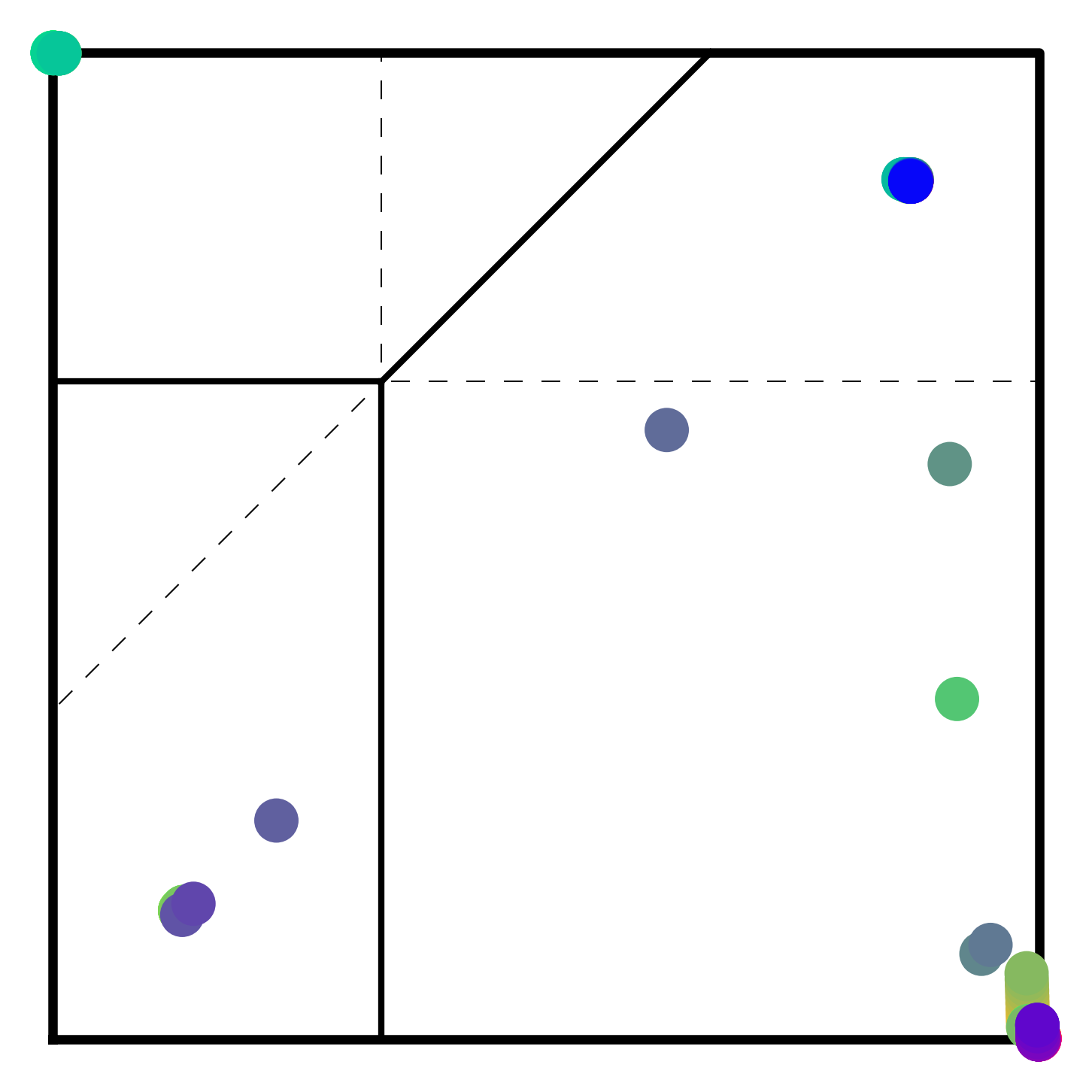}
\includegraphics[width=.32\linewidth]{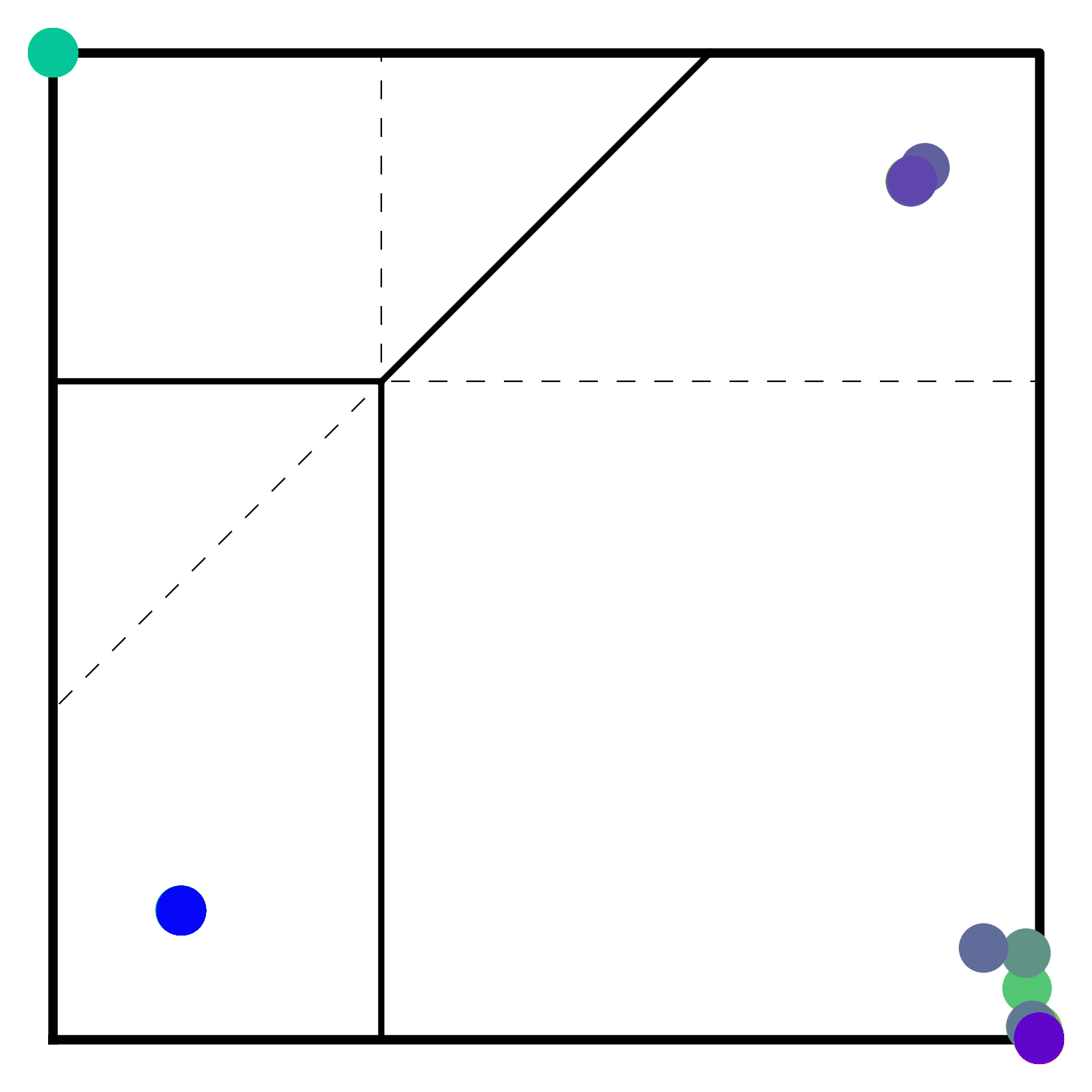}
\includegraphics[width=.32\linewidth]{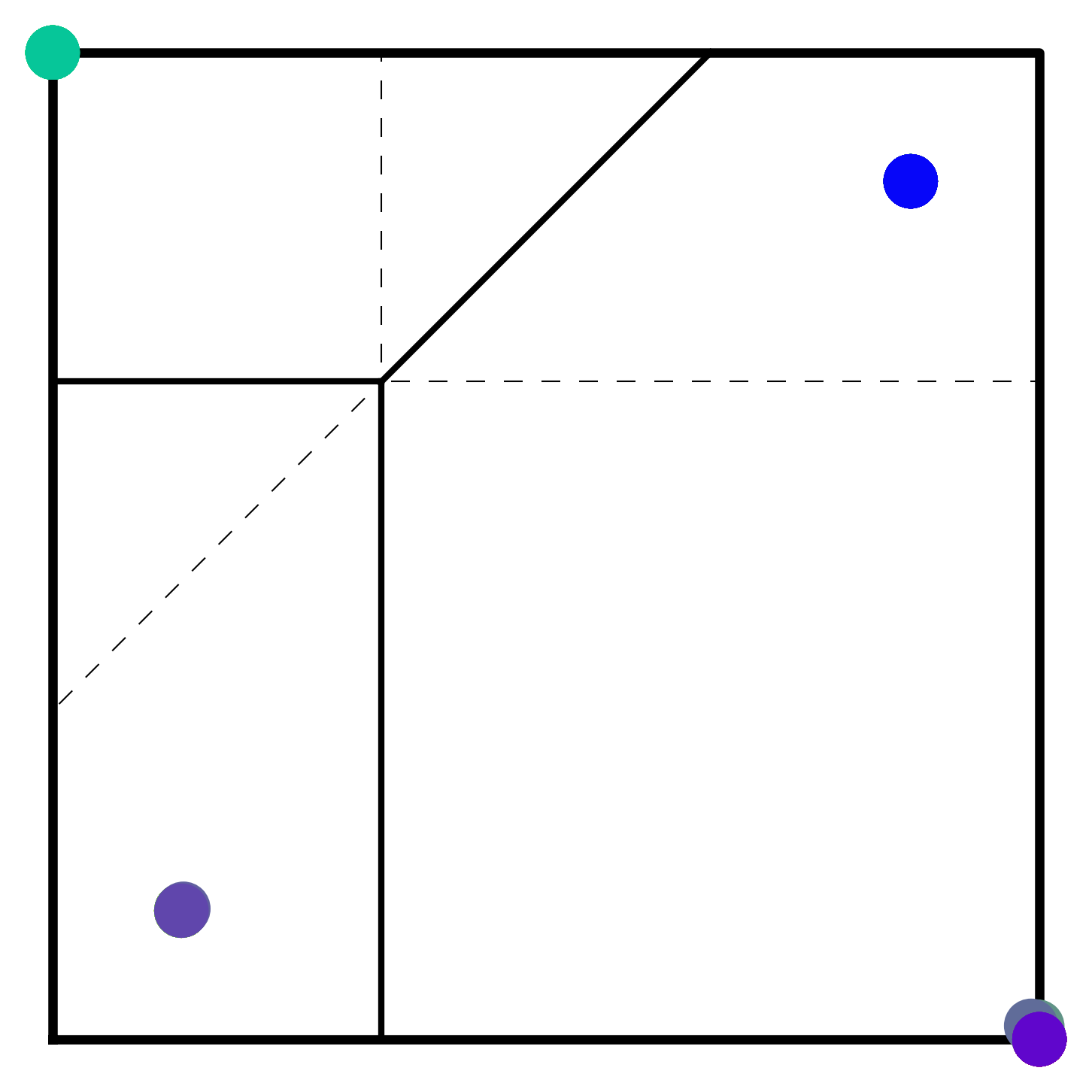}
\caption{An example of a CS Polling Dynamics obtained by perturbation of the Polling Dynamics used in the proof of Theorem \ref{t:example2}, showing robustness of the bad $2$-cycle. Top-left: initial points of the plane; then from left to right then top to bottom, the first five iterations of a continuous CPD where, whenever margins are above $4\,\%$, $85\,\%$ of voters apply the Leader Rule and $15\,\%$ keep their ballot unchanged. Points are drawn with the color corresponding to initial position (darker above for better readability). After the first iteration, already most points are attracted to the periodic points near the four corners. The top-left and bottom-right corners are attracting fixed points, and there is an attractive orbit of period $2$ near the other corners.}
\label{f:continuous-numeric}
\end{center}\end{figure}

The region $A_1\subset A_{abc}$ delimited by the lines of equations $(z<x+\frac16)$ and $(z>\frac56)$ results in the outcome $abc$ with margins of $\frac{1}{24}$th of the electorate, i.e. slightly over $4\%$. Similarly, the region $A_2\subset A_{cab}$ delimited by the lines of equations $(z<x+\frac16)$ and $(x<\frac16)$ result in the outcome $cab$ with the same margins.
We have $\Phi_{.15}(x,z)=(.15 x, .15 z)$ whenever $(x,z)\in A_1$ and $\Phi_{.15}(x,z)=(.85+.15 x,.85+.15 z)$ whenever $(x,z)\in A_2$. One easily checks that $\Phi_{.15}(A_1)\subset A_2$ and $\Phi_{.15}(A_2)\subset A_1$. It follows that $\Phi_{.15}^2(A_1)\subset A_1$, and since $\Phi_{.15}^2$ is a contraction on $A_1$ (of ratio $.15^2$, much small than $1$), the Contraction Mapping Theorem ensures it has a fixed point $(x_1,z_1)\in A_1$. Then the orbit of $(x_1,z_1)$ is a $2$-cycle with one state inducing the outcome $abc$ and the other inducing $cab$. Moreover any element of the open set $A_1\cup A_2$ is attracted to this cycle exponentially fast.
\end{exem}

\subsection{Numerical study}\label{s:experimental}

The examples of the previous section are interesting from a theoretical point of view, showing that the existence of equilibria and the quality of candidates they elect do not by themselves suffice to ensure a good outcome. But from a practical point of view, it could be that the bad cycles we exhibited are so rare that they do not matter too much.
In this section, we present experimental estimates of the frequency of bad cycle in preference profiles where a Condorcet winner exists. We continue to follow the setup of Sections \ref{s:example1} and \ref{s:example2}, using either the Leader Rule or the modified Leader Rule. This leaves us to model the preference profile, and for this we compare several ``cultures''.

Our results complement for example the findings of  \cite{lepelley2000computer} and \cite{laslier2010silico}. Compared to the latter, for each preference profile we determine the existence of a bad cycle anywhere in the space of outcomes, and of arbitrary length.

\paragraph*{Main experimental setup}

We fix the number of candidates to $m=6$, and the number of voter types to $n=20$, $N=\{1,\dots,20\}$, a compromise between complexity of the preference profile and computational power limitations. The influence of these numbers is tested at the end of the section.

Each voter type is assigned a number of voters uniformly drawn in $[0,1]$ and is assigned a preference modeled by a linear order on $\bar A = A\cup\{\limit\}$; when we use the Leader Rule, $\limit$ is simply ignored, but when we use its modified version, the candidates ranked beyond $\limit$ are considered tied for the last rank and thus never put on this voters' ballots. For example, the order $a\spref b\spref\limit\spref c\spref d\spref e\spref f$ corresponds to the preference $abcdef$ when the LR is used, and to $ab(cdef)$ when its modified version is used.

We consider the following ways of constructing preferences:
\begin{itemize}
\item \emph{Impartial Culture}: each voter type draws uniformly an element $\mathcal{L}(\bar A)$,
\item \emph{$d$-dimensional culture}: each candidate and each voter types is given a uniform random position in $[0,1]^d$ where each coordinate represents a ``political axis''; the preferences of a voter type is obtained by sorting the candidates by increasing distance to the voter type, using the $\ell^1$ metric
\[ d((x_1,\dots, x_d),(y_1,\dots,y_d)) = \sum_{i=1}^d \lvert x_i-y_i\rvert\]
that sums the levels of disagreement along the various axes; $\limit$ is then inserted by drawing a random ``distance'' to the voter type, with the same law than the distance between two uniform points in $[0,1]^d$
\end{itemize}

We ran experiments for each combination of a heuristic among the Leader Rule and its modification and each political culture among the Impartial Culture and the $d$-dimensional cultures with $d\in\{1,2,3,400\}$. In each case, we generated $100\,000$ independent preference profiles and recorded whether a Condorcet winner exists, and if yes whether the Polling Dynamics has a ``bad cycle'' (or bad equilibrium), i.e. where some of the ballot profiles do not elect the Condorcet winner. Each run took a couple of hours on a single core of a modern CPU. The results are provided in Table \ref{tab:experiments}.
The first take-away is that in practice
\begin{important}
The Leader Rule is very effective in electing the Condorcet winner.
\end{important}
Only the Impartial Culture witnesses slightly larger odds of not electing her while structured cultures very rarely produce bad cycles, the linear ($d=1$) culture having produced none in $100\,000$ attempts (this is no accident, see Section \ref{s:1d}). 

\begin{table}[ht]
\begin{center}
\begin{tabular}{|l||p{1.6cm}|p{1.6cm}|p{1.6cm}|p{1.6cm}|p{1.6cm}|}
\hline
\diagbox{Heuristic}{Culture} & Impartial& $d=1$ & $d=2$ & $d=3$ & $d=400$ \\
\hline
\hline
Leader Rule  & {\scriptsize($70\%$)}  $1.3\%$ & {\scriptsize($100\%$)}  $0.0\%$ & {\scriptsize($90\%$)}  $0.2\%$ & {\scriptsize($87\%$)}  $0.2\%$  & {\scriptsize($85\%$)}  $0.3\%$ \\
\hline
Modified Leader Rule & {\scriptsize($75\%$)}  $6.3\%$  &{\scriptsize($92\%$)} \hspace{1ex}  $15\%$   & {\scriptsize($89\%$)}  $7.8\%$  & {\scriptsize($88\%$)}  $6.0\%$  & {\scriptsize($87\%$)}  $3.0\%$ \\
\hline
\end{tabular}
\caption{Experimental results with $6$ candidates and $20$ voter types. In small, the proportion of preference profiles where a Condorcet winner exists; in normal size, the proportion of preference profiles with a bad cycle or equilibrium, among preference profiles having a Condorcet winner (both rounded). The $95\%$ confidence Wilson score interval gives a deviation of less than $0.23$ percentage point for all values; the deviation is even less than $0.04$ percentage point for the value $0.3\%$ and below.}\label{tab:experiments}
\end{center}
\end{table}

On the contrary, when voters apply the Modified Leader Rule and thus, with the kind of preferences used here, refuse to approve of candidates below their personal threshold no matter what, bad cycles are more common. They never dominate, but they culminate \emph{precisely} in the case of a linear culture, with $15\%$ of preference profiles leading to a bad cycle in the Polling Dynamics. The situation that is most favorable to the Leader Rule, is also the worst one for its modification! 

Imagine that voters of a given type are aware of a Condorcet winner quite low on their preferences, and that they know that when everyone uses the Leader Rule, the election of the Condorcet winner is very likely. By not applying the Leader Rule and instead refuse to approve candidates below a threshold, they can get significant odds to obtain the election of a candidate they prefer to the Condorcet winner:
\begin{important}
When some voters expect a candidate they do not like to be a Condorcet winner, their best interest can be not to apply the Leader Rule, at least not below a certain threshold.
\end{important}
At first sight, this may seem difficult to reconcile with the optimality of the Leader Rule proven by Laslier \cite{Laslier2009leader}. The point is that Laslier considers that each voter assumes imperfect recording of ballots and think as if she would be the only one to change her ballot, trying to optimize the expected outcome. Our numerical results show that whenever voters expect other voters of the same type to also change their votes strategically, the Leader Rule might not be their best heuristic.

This experiment also gave the opportunity to find especially problematic Polling Dynamics, see Figures \ref{f:experiments-relaxed-2} and \ref{f:experiments-MLR-d1} where green outcomes correspond to the Condorcet winner being elected, orange outcomes are part of a cycle but do not elect the Condorcet winner, and red outcomes are equilibria not electing the Condorcet winner.

%
%

\begin{figure}[htp]
\centering
\includegraphics[width=\linewidth]{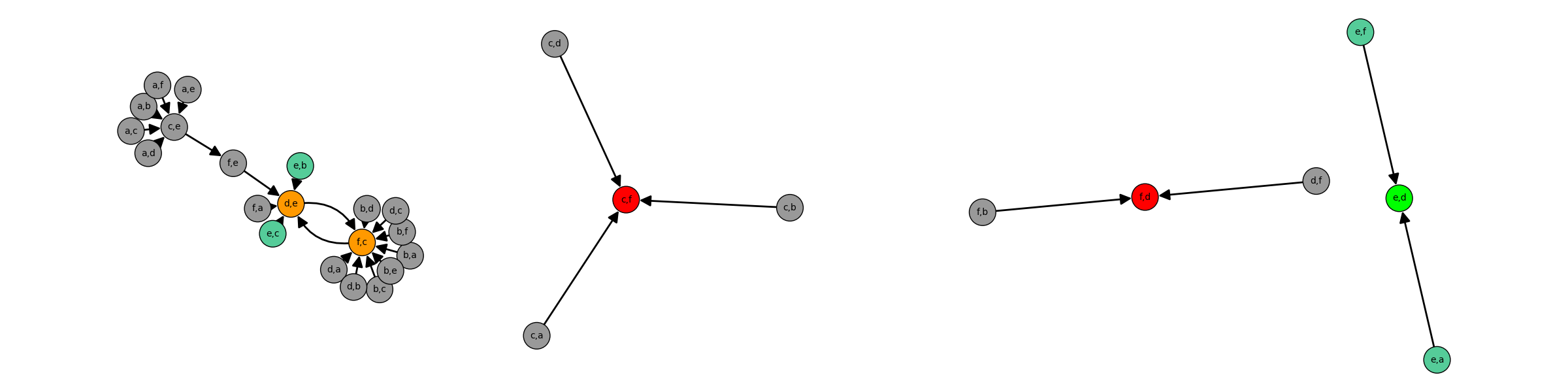}
\caption{A bad cycle and two equilibria not electing the Condorcet winner under the modified Leader Rule, produced with the Impartial Culture. Only $3$ outcomes out of $30$ end up electing the Condorcet winner, and there are outcomes where the Condorcet winner is elected, but that lie in the basin of attraction of the bad cycle (this never happens with the Leader Rule).}\label{f:experiments-relaxed-2}
\end{figure}

\begin{figure}[htp]
\centering
\includegraphics[width=\linewidth]{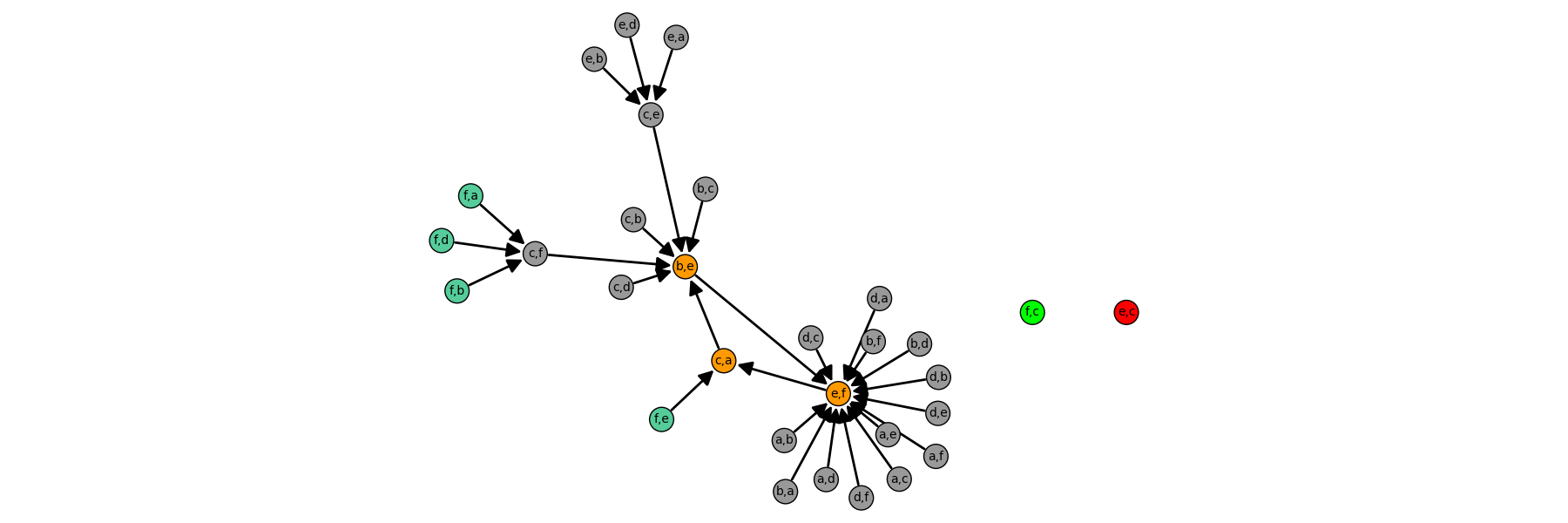}
\caption{An example in the one-dimensional culture with the modified Leader Rule: the bad cycle attracts all outcomes but two, including most outcomes electing the Condorcet winner.}\label{f:experiments-MLR-d1}
\end{figure}

\paragraph*{Effect of the numbers of candidates and voter types}

We report in Table \ref{tab:experiments2} more modest numerical experiment to assert the influence of the number of candidates and the number of voters.
We see that the number of voter types has relatively small influence, at least in the range tested, with the exception of the $1$-dimensional culture with the modified LR, where a larger voter type makes bad cycles and equilibria  more likely. The main take away, valid in all considered conditions except the $1$-dimensional culture with the LR, is that:
\begin{important}
A larger number of candidates seems to make bad cycles and equilibria more common among situations where a Condorcet Winner exists.
\end{important}
Computational power prevented us to explore larger number of candidates with samples large enough to draw conclusions.

\begin{table}[htp]
\begin{center}{\footnotesize
\begin{tabular}{|l||p{1.3cm}|p{1.3cm}|p{1.3cm}|p{1.3cm}|p{1.3cm}|p{1.3cm}|}
\hline
\diagbox{Cul. \\\& Strat.}{$\lvert\mathcal{C}\rvert$ \&  $\lvert\mathcal{T}\rvert$} & $3$, $10$ & $4$, $10$ & $6$,  $10$ & $\boldsymbol{6}$, $\boldsymbol{20}$ & $6$, $30$ & $8$, $20$  \\
\hline
\hline
Impartial, LR  & {\scriptsize ($92\%$)} $0\phantom{\%}$  & {\scriptsize ($84\%$)} $0.3\%$  &  {\scriptsize ($71\%$)} $1.0\%$  & {\scriptsize ($\boldsymbol{70\%}$)} $\boldsymbol{1.3\%}$ & {\scriptsize ($69\%$)} $1.3\%$  &   {\scriptsize ($60\%$)} $2.0\%$ \\
\hline
$1$-dim., LR &{\scriptsize ($100\%$)} $0\phantom{\%}$ &{\scriptsize ($100\%$)} $0\phantom{\%}$  & {\scriptsize ($100\%$)} $0\phantom{\%}$ & {\scriptsize ($\boldsymbol{100\%}$)} $\boldsymbol{0}$ & {\scriptsize ($100\%$)} $0\phantom{\%}$  & {\scriptsize ($100\%$)} $0\phantom{\%}$\\
\hline
$2$-dim., LR  &{\scriptsize ($97\%$)} $0\phantom{\%}$ &{\scriptsize ($94\%$)} $0.05\%$  & {\scriptsize ($90\%$)} $0.3\%$ & {\scriptsize ($\boldsymbol{90\%}$)} $\boldsymbol{0.2\%}$ & {\scriptsize ($91\%$)} $0.05\%$  & {\scriptsize ($87\%$)} $0.5\%$\\
\hline
$3$-dim., LR  &{\scriptsize ($96\%$)} $0\phantom{\%}$ &{\scriptsize ($92\%$)} $0.1\%$  & {\scriptsize ($85\%$)} $0.4\%$ & {\scriptsize ($\boldsymbol{87\%}$)} $\boldsymbol{0.2}\%$ & {\scriptsize ($89\%$)} $0.1\%$  & {\scriptsize ($84\%$)} $0.4\%$\\
\hline
$400$-dim., LR  &{\scriptsize ($95\%$)} $0\phantom{\%}$ &{\scriptsize ($90\%$)} $0.1\%$  & {\scriptsize ($81\%$)} $0.3\%$ & {\scriptsize ($\boldsymbol{85\%}$)} $\boldsymbol{0.3\%}$ & {\scriptsize ($89\%$)} $0.3\%$  & {\scriptsize ($81\%$)} $0.6\%$\\
\hline
Impartial, modified LR &{\scriptsize ($94\%$)} $2,1\%$ &{\scriptsize ($88\%$)} $3.5\%$  & {\scriptsize ($77\%$)} $5,9\%$ & {\scriptsize ($\boldsymbol{75\%}$)} $\boldsymbol{6.3\%}$ & {\scriptsize ($75\%$)} $6.4\%$  & {\scriptsize ($66\%$)} $8.0\%$\\
\hline
$1$-dim., modified LR &{\scriptsize ($98\%$)} $2.6\%$ &{\scriptsize ($96\%$)} $5.4\%$  & {\scriptsize ($90\%$)} $10\%$ & {\scriptsize ($\boldsymbol{92\%}$)} $\boldsymbol{15\%}$ & {\scriptsize ($92\%$)} $19\%$  & {\scriptsize ($88\%$)} $20\%$\\
\hline
$2$-dim., modified LR &{\scriptsize ($97\%$)} $1.5\%$ &{\scriptsize ($95\%$)} $3.4\%$  & {\scriptsize ($89\%$)} $6.2\%$ & {\scriptsize ($\boldsymbol{89\%}$)} $\boldsymbol{7.8\%}$ & {\scriptsize ($91\%$)} $7.8\%$  & {\scriptsize ($84\%$)} $11\%$\\
\hline
$3$-dim., modified LR &{\scriptsize ($97\%$)} $1.4\%$ &{\scriptsize ($93\%$)} $2.7\%$  & {\scriptsize ($86\%$)} $5.3\%$ & {\scriptsize ($\boldsymbol{88\%}$)} $\boldsymbol{6\%}$ & {\scriptsize ($90\%$)} $5.9\%$  & {\scriptsize ($84\%$)} $8.2\%$\\
\hline
$400$-dim., modified LR &{\scriptsize ($96\%$)} $0.9\%$ &{\scriptsize ($92\%$)} $2.1\%$  & {\scriptsize ($84\%$)} $3.4\%$ & {\scriptsize ($\boldsymbol{87\%}$)} $\boldsymbol{3.0\%}$ & {\scriptsize ($89\%$)} $2.6\%$  & {\scriptsize ($82\%$)} $4.1\%$\\
\hline
\end{tabular}}
\caption{Complementary experimental results to check robustness. Samples have $6\,250$ pseudo-random preference profiles (except for $6$ candidates and $20$ voter types, where we reported the previous results with larger samples), yielding typically $4$ times larger confidence intervals than samples of size $100,000$ and leaving the expected deviations under the percentage point. In small, the proportion of preference profiles where a Condorcet winner exists; in normal size, the proportion of preference profiles with a bad cycle or equilibrium, among profiles having a Condorcet winner (both rounded).}\label{tab:experiments2}
\end{center}
\end{table}

\subsection{Convergence of the Leader Rule in one-dimensional cultures}\label{s:1d}

In view of the above experiments, it seems that the one-dimensional culture combines well with the Leader Rule. Black's Median Voter Theorem \cite{black1958theory} implies the existence of a Condorcet winner, and from \cite{Laslier2009leader} it follows that there exist equilibria and that all equilibria elect $\omega$. Here we prove that in this particular culture, we moreover have convergence of the Polling Dynamics.

\begin{defi}
We say that a preference profile \emph{can be modeled by a one-dimensional culture}
when there exist a \emph{positional mapping}
\[ x : A\cup N \to \mathbb{R} \]
with the following property: the preferences $\pref$ of any voter type $i\in N$ is given by
\[\forall \alpha\neq\beta\in A,\qquad \alpha \pref_i \beta \quad\text{if and only if}\quad \lvert x(\alpha) - x(i)\rvert < \lvert x(\beta) - x(i)\rvert.\]
For simplicity, we assume further that $\lvert x(\alpha) - x(i)\rvert \neq \lvert x(\beta) - x(i)\rvert$ for all $\alpha\neq\beta$ and all $i$, so that all preferences are linear orders (without ties); and we assume that there is no partition of the voter types in two groups that have equal total number of voters, so that no duel between two candidates would end up in a tie (these are \emph{generic} conditions: they are stable under small perturbation of the data and any data can be approximated arbitrarily close by data satisfying those conditions).
\end{defi}

The goal of this section is to prove the following result.
\begin{theo}\label{t:1d-convergence}
Under Approval Voting, if the preference profile can be modeled by a one-dimensional culture and voters apply the Leader Rule, then the Polling Dynamics converges to an equilibrium (which elects the Condorcet winner, as prescribed by Laslier's Theorem).
\end{theo}


The end of the section is dedicated to the proof of Theorem \ref{t:1d-convergence}. We use outcomes reduced to the winner and runner-up, we assume that the preference profile can be modeled by a one-dimensional culture and we fix a positional mapping $x$. Denoting by $w= \sum_{i\in N} w_i$ the total number of voters, a \emph{median} is a value $m\in\mathbb{R}$ such that
\[ \sum_{i\colon x(i)\le m} w_i \ge \frac w2 \qquad\text{and}\qquad \sum_{i\colon x(i)\ge m} w_i \ge \frac w2\]
i.e., at least half the voters lie on the left of $m$, and at least half have lie on the right of $m$. The assumption that there is no partition of voter types in two groups of equal size ensures that there is a unique median $m$, coinciding with the position of some voter type: $m=x(i_m)$, where the voter type $i_m$ is called the \emph{median type}. By assumption, there is a single candidate $\mu$ whose position is closest to $x(i_m)$.

We start with the following particular instance of Black's Median Voter Theorem, which we prove for the sake of completeness.
\begin{lemm}\label{l:1d-winner}
The candidate $\mu$ is a Condorcet winner. 
\end{lemm}

\begin{proof}
Let $\alpha$ be any other candidate. Voters preferring $\mu$ to $\alpha$ are those positioned on the half line $L$ of endpoint $\frac12(x(\alpha)+x(\mu))$ and containing $x(\mu)$. Since $x(\mu)$ is closer from $x(X_m)$ than $x(\alpha)$, $L$ contains $x(i_m)$ and thus contains at least half the voters. Since there are no possible ties, $\mu$ dominates $\alpha$.
\end{proof}

\begin{lemm}
Let $\alpha\beta$ be an outcome such that $\alpha\neq\mu$ and let $\alpha'\beta'=\psi(\alpha\beta)$ be the next outcome in the (shifted) Polling Dynamics. Then either $\lvert x(\alpha')-x(\mu)\rvert < \lvert x(\alpha)-x(\mu)\rvert$, or $\alpha' = \alpha$ and $\lvert x(\beta')-x(\mu)\rvert < \lvert x(\beta)-x(\mu)\rvert$.
\end{lemm}

\begin{proof}
The positions $x(\alpha)$ and $x(\beta)$ divide the real line in three component: the open bounded interval $I$ between them, the open half-line $H_\alpha$ with extremity $x(\alpha)$ avoiding $I$, and the open half-line $H_\beta$ with extremity $x(\beta)$ avoiding $I$. 

If $x(\mu)\in H_\alpha$, $\alpha$ will receive the votes of all voters positioned in a half line starting at $(x(\alpha)+x(\beta))/2$ and containing $x(\mu)$, while other candidates receive votes only from voters positioned in $H_\alpha$ or in $I\cup H_\beta$. It follows that $\alpha' = \alpha$ and $\beta'$ is the candidate positioned next to $\alpha$, in the direction of $x(\mu)$ (possibly $\beta'=\mu$). In this case, we thus have $\alpha' = \alpha$ and $\lvert x(\beta')-x(\mu)\rvert < \lvert x(\beta)-x(\mu)\rvert$.

In all other cases, $\alpha'$ is the candidate positioned next to $\alpha$, in the direction of $\mu$ (possibly $\alpha'=\mu$), therefore $\lvert x(\alpha')-x(\mu)\rvert < \lvert x(\alpha)-x(\mu)\rvert$.
\end{proof}

Let $>$ denote the ``geo-lexicographic'' order on outcomes, defined by 
$\alpha\beta>\alpha'\beta'$ whenever
\[\lvert x(\alpha')-x(\mu)\rvert < \lvert x(\alpha)-x(\mu)\rvert, \text{ or }\alpha' = \alpha \text{ and }\lvert x(\beta')-x(\mu)\rvert < \lvert x(\beta)-x(\mu)\rvert.\]
The previous lemma shows that along an orbit of $\psi$, the outcome can only decrease in this order until $\mu$ becomes winner. Since there cannot be an infinite decreasing sequence of outcomes, eventually $\mu$ becomes winner. Now, as is well-known, this is a stable situation: for all $\beta$, $\psi(\mu\beta)=\mu\beta_0$ where $\beta_0$ is the candidate getting the better score in a duel against $\mu$. Therefore, every outcome converges under $\psi$ to the outcome $\mu\beta_0$ (and every ballot profile of $B^n$ converges under $\varphi$ to the ballot profile induced by the outcome $\mu\beta_0$ and the Leader Rule). The proof of Theorem \ref{t:1d-convergence} is complete.

\section{Other Voting rules}\label{s:other}

In order to show how our framework applies to general voting rule, let us consider three further examples: 
one for Plurality voting, where contrary to Example \ref{ex:cycle} heuristics use quantitative information in the outcome rather than only the ranking;
one for Instant Run-Off voting (IRV);
one for the \emph{condorcified} version of IRV (or, for that matter, of any ranked voting). 

There is a paucity of examples in the literature concerning \emph{synchronous} iterative voting; we can still mention Example 1 in \cite{chopra2004knowledge} (whose heuristics imply that voters know the number of voters of the same type), and Example 7 in \cite{Brams2007approval} (which needs a change of heuristic along the cycle).

\subsection{A plurality cycle}

We consider Plurality voting, with outcomes giving the total number of votes of each candidate, with $n=4$ voter types and $m=3$ candidates, $A=\{a,b,c\}$. We thus set $\mathcal{O}=[0,w]^A$ where $w=\sum_{i\in N} w_i$ is the total number of voters, an element of which is written as a triple $r=(r_a,r_b, r_c)$. Ballots bear a single name, i.e. $\mathcal{B}=A$, and the information function $g$ sends a family $(B_i)_{i\in N}\in A^N$ of ballots to the outcome
\[g((B_i)_{i\in N}) = \Big( \sum_{i\in N, B_i=\alpha} w_i  \Big)_{\alpha\in A.}\]
The winner map sends $(r_a,r_b,r_c)$ to the candidate with the most votes, with ties broken in favor of the earlier candidate in the alphabetical order.

We consider the preference profile
\[\begin{array}{ccc}
X   & Y   & Z   \\
10 & 13 & 14 \\
\hline
a   & b   & c   \\
b   & a   & ab  \\
c   & c   & 
\end{array}\]
and assume the following heuristics:
\begin{itemize}
\item voters of type $Z$ always vote for $c$,
\item voters of type $X$ vote as follows. If their least-preferred candidate $c$ is in a close contest for winning with their second-preferred candidate $b$, more precisely if
\[\lvert r_c - r_b \rvert \le .03(r_a+r_b+r_c) \quad\text{and}\quad \min(r_b,r_c)>r_a+.03(r_a+r_b+r_c)\]
then they vote for $b$; in all other cases they vote sincerely, for $a$,
\item voters of type $Y$ vote as voters of type $X$, with $a$ and $b$ switched.
\end{itemize}
These heuristics allow voters of type $X$ and $Y$ to vote strategically whenever their preferred choice is far from winning, but their second-preferred choice needs support to avoid the win of their least-preferred candidate. The factor $.03$ is chosen because it is the typical uncertainty in polls on about a thousand voters. With the above preference profile, $.03(r_a+r_b+r_c)= 1.11$.

If voters first report their preferred candidates, we obtain the outcome
\[r_0=(10, 13, 14)\]
and we are in the situation where $c$ wins in close contest with $b$. The above heuristics lead voters of type $X$ to vote for $b$, while others remain sincere, leading to the new outcome $r_1=\psi(r_0) = (0, 23, 14)$ where $b$ wins by a large margin. The above heuristics lead every voter to vote sincerely, so that $\psi(r_1)=r_0$, and we have a two-cycle (an outcome of which elects the absolute majority loser). By Theorem \ref{t:robust}, this bad cycle persists in a small enough CS perturbation.

\subsection{An IRV cycle}

We now consider IRV and give an example with $n=3$ and $m=4$.
Ballots and outcomes are linear rankings of candidates: $\mathcal{B} = \mathcal{O} = \mathcal{L}(A)$. The information function $g$ is defined as follows. Given an argument $(B_i)_{i\in N}\in \mathcal{B}^N$, we first rank candidates according to the total weight of ballots where they are ranked top (again, ties broke in alphabetical order). The last ranked candidate is eliminated, and we rank the remaining candidates according to the total weight of ballots where they are ranked top \emph{among candidates that have not been eliminated}. The last ranked candidate is again eliminated. The resulting outcome is denoted with the last standing candidate first, then the other by reversed order of elimination. The winner map sends a ranking $r\in \mathcal{O} = \mathcal{L}(A)$ to its top candidate.

We consider the preference profile
\[\begin{array}{ccc}
X   & Y   & Z   \\
10 & 11 & 12 \\
\hline
b   & c   & d   \\
a   & a   & a   \\
c   & b   & b   \\
d   & d   & c
\end{array}\]
so that the Condorcet order coincides with the alphabetical order, in particular $a$ is the Condorcet winner. Let us define simple heuristics that make use of an expected outcome to try to improve it. Denote by $\alpha\beta\gamma\delta$ the preferences of the considered type $i$ of voters (e.g. for $i=X$, $\alpha=b$, $\beta=a$, $\gamma=c$ and $\delta=d$). Given an outcome $r$, if the expected winner is neither $\alpha$ nor $\beta$ and $\beta$ was eliminated before $\alpha$ then $\sigma_i(r)=\beta\alpha\gamma\delta$; for every other outcome $r$, $\sigma_i(r)=\alpha\beta\gamma\delta$. In other words, voters strategically invert their two first preferred candidates when the top one had a better run but still could not prevent a worse candidate to be elected, giving their second choice a chance to do better.

If voters start voting sincerely, we obtain the outcome $cdba$ with the Condorcet winner eliminated with no votes at all in the first round, and ballots of type $X$ voters ultimately transferred to candidate $c$ against $d$. The above heuristics have voters of type $X$ and $Z$ invert their two most preferred candidates and vote $abcd$, $adbc$ respectively. The outcome is then $acbd$. Then all voters resume voting sincerely, and we have a $2$-cycle one of whose ballot profiles elects the third candidate in the Condorcet order. Observe that every outcome where the Condorcet winner is elected leads voters to cast their sincere ballot, thus leading to this bad cycle.

According to Theorem \ref{t:robust}, this cycle is again robust under perturbation in a CS setting.

\subsection{A cycle in condorcified ranked voting}

Consider a ranked voting rule, i.e. one with $\mathcal{B}=\mathcal{L}(A)$. It has been proved by Durand, Mathieu and Noirie \cite{durand2014condorcification,durand2016condorcet} and independently by Green-Armytage, Tideman and Cosman \cite{green2016statistical} that under mild conditions \emph{condorcification} of a voting system (i.e. electing the Condorcet winner if she exist and applying the given rule otherwise) cannot increase manipulability and often reduces it. We can therefore ask whether a condorcified voting rule is less susceptible to the presence of bad cycles when a Condorcet winner is present. Unsurprisingly, the answer is negative when voters who do not like the Condorcet winner strategically choose to rank her lower in their ballot

Let us consider a condorcified voting rule, i.e. for all $(B_i)_{i\in N}\in \mathcal{B}^N$ for which their is a Condorcet winner $\gamma$, $W\circ g((B_i)_{i\in N})=\gamma$. Assume moreover that $\mathcal{O}$ contains the information whether or not the winner was a Condorcet winner (with respect to the \emph{ballot} profile, since the preference profile is not available information). Consider the preference profile
\[\begin{array}{cccc}
 Z   & Y   & X   & W   \\
150  & 102 & 101 & 100 \\
\hline
a    & b   & c   & b   \\
b    & c   & a   & a   \\
c    & a   & b   & c
\end{array}\]
for which $a$ is the sincere Condorcet winner. Assume the following heuristics for voter of any type $i$, whose preferences are denoted as above by $\alpha\beta\gamma$: if $\beta$ is elected as Condorcet winner, then they cast the ballot $\alpha\gamma\beta$, in an attempt to give $\alpha$ a shot. Otherwise, they vote sincerely.

If at first all voters vote sincerely, in the first outcome $r_0$ the winner is $a$, declared a Condorcet winner. This leads voters of type $X$ to next cast the ballot $cba$ and voters of type $W$ to cast the ballot $bca$. This makes $b$ the Condorcet winner of the new ballot profile. The above heuristics then lead all voters to vote sincerely, except voters of type $Z$ who cast the ballot $acb$; this makes again $a$ the Condorcet 
winner, and we have a $2$-cycle, one of whose ballot profiles elects a candidate that is not the Condorcet winner of the preference profile. Moreover this cycle attracts all outcomes where $a$ is a Condorcet winner of the ballot profile.

Again, Theorem \ref{t:robust} ensures that this cycle is robust under perturbation in a CS setting.

\section{Chaos}\label{s:chaos}

We finish with an example illustrating the flexibility of the CS setting for the Polling Dynamics. The starting point is to observe that the embedding $\Phi_0$ of a discrete space Polling Dynamics in a CS setting cannot be continuous (unless it is a constant map). Indeed, it takes a finite number of values, and the state space $\mathcal{P}$ is connected. There must exist some CS ballot profiles near which arbitrarily close CS ballot profiles are sent to very different images by $\Phi_0$. In practice, assuming the outcomes carry continuous information (e.g. shares of votes for each candidate), a heuristic like the Leader Rule would probably not be applied blindly by all voters when the expected winner and runner-off are in a close-call contest, or when the runner-up is in a close-call contest with the next candidate. Different voters will have different confidence in the available information and a small change in the expected outcome should result in a small change in the ballots cast, i.e. the CS Polling Dynamics should be continuous.

We now give a relatively simple example with $m=3$, based on a reluctance to add one's second-preferred candidate unless it seems likely to improve the outcome, and we observe that a complicated dynamics emerges. The take-away is that
\begin{important}%
Even simple CS Polling Dynamics can exhibit a chaotic behavior, where the sequence of winners is impossible to predict reliably from the observation of arbitrarily many of its first terms.
\end{important}
By looking at the influence of a small number of core parameters of the model, we shall see that chaos is neither universal nor restricted to exceptional parameters.

\paragraph*{A simple continuous CS Polling Dynamics}
\label{s:chaos-example}

We consider the four-types preference profile used in the proof of Theorem \ref{t:example2}:
\[\begin{array}{cccc}
Z   & Y   & X   & W   \\
3 & 1   & 3 & 5 \\
\hline
a   & a   & b   & c   \\
b   & bc & a   & ab\\
c   &     & c   &   
\end{array}\]
for which $a$ is a Condorcet winner and $c$ an absolute majority loser.
As above, the corresponding sets of admissible ballots makes $\mathcal{P}$ a square with coordinates $(x,z)$ where $x$, $z$ are the proportions of voters of type $X$ , $Z$ respectively voting $\{a,b\}$. This choice of ballot can be seen either as a form of cooperation between the two types of voters, or as the result of risk aversion (which are complementary views, not opposed ones).

An outcome is a triple $r=(r_a,r_b,r_c)\in[0,1]^3$ giving the shares of votes obtained by each candidate, and we choose the following CS heuristics (here the resulting distribution of ballots does not depend upon the value $x$ or $z$ in $\Delta(\mathcal{B}_i)$ that is by definition part of their argument):
\[\sigma_Z((r_a,r_b,r_c),z) = C_Z\circ S(r_a,r_b,r_c) \qquad \sigma_X(x,(r_a,r_b,r_c),x)= C_X\circ S(r_b,r_a,r_c) \]
where $C_Z$, $C_X$, $S$ will be defined below, to be interpreted as follows:
\begin{itemize}
\item $S$ is a ``safety function'', quantifying how unlikely it seems that cooperating would be useful to counter a threat by $c$ (notice how $r_a, r_b$ are exchanged in $\sigma_X$, to take into account the preferences of this type). $S$ will be in particular very small when both $c$ is expected winner or close to be expected winner, \emph{and} collaborating has a good chance to prevent her to win,
\item $C$ is a ``collaboration function'', translating a level of safety into a proportion of collaborations (high safety resulting in low collaboration).
\end{itemize}
We take here a very simple safety function:
\[S(r_1,r_2,r_3) = \begin{cases}
\lvert r_2-r_3\rvert &\text{when } r_2 > r_1 \\
\frac12 \lvert r_2-r_3\rvert + \frac12 \lvert r_1 - r_3\rvert &\text{otherwise.}
\end{cases} \]
i.e. when the second-preferred candidate is ranked higher than the preferred one, the safety is the margin in her race with the least-preferred one; otherwise, safety is an average of the margins in races of the two preferred candidates against the last one. For example, if $r_1\gg r_3\simeq r_2$ then the safety is not too small (the least preferred candidate has little chance of winning) but not maximal, meaning some voters will prefer to collaborate in order to ensure $c$ finishes at the last rank.  When $r_1\simeq r_2$, the safety is close to both $\lvert r_1-r_3\rvert$ and $\lvert r_2-r_3\rvert$; when $c$ is far above or far below, one's vote seems unlikely to change the outcome and the incentive to collaborate is small.

Finally, we choose a simple collaboration function extrapolating linearly between the value $1$ (all voters collaborate) when the safety vanishes, and the value $0$ (no voter collaborates) when the safety is large enough:
\[C_i(t) = (1-c_i t)_+ := \max(0,1-c_i t)\]
where as indicated $(\cdot)_+$ is the positive part, and where the coefficient $c_i$ ($i= Z, X$) quantifies the risk tolerance: the higher, the less likely $X$ and $Z$ voters will collaborate (low values of $c_i$ thus correspond to high aversion to risk). For now, we take $c_Z=c_X=5$, a rather moderate value: $1-5t$ reaches $0$ only at $t=0.2$, i.e. when the safety margins reaches a staggeringly high $20\%$ of voters.

Since a state $(x,z)$ results in the outcome $(r_a,r_b,r_c)=(3x+4, 3z+x, 5)$, these choices yield the CS Polling Dynamics $\Phi$ defined by:
\[\Phi(x,z) = \big(\big(1-5\lvert 3x-1\rvert\big)_+, \big(1-2.5\lvert 3x-1\rvert - 2.5\lvert 3z+x-5\rvert\big)_+  \big)\]
when $2x+4\ge 3z$ and
\[\Phi(x,z) = \big( \big(1-2.5\lvert 3z+x-5\rvert - 2.5 \lvert 3x-1\rvert\big)_+, \big(1-5\lvert 3z+x-5\rvert\big)_+  \big)\]
when $2x+4\le 3z$. While $\Phi$ may look like an innocent map, drawing an orbit of this map reveals an interesting pattern (Figure \ref{f:chaos}). It turns out \emph{all} orbits yield pretty much the same image, typical of a \emph{chaotic attractor}. 

In dynamical systems, \emph{chaos} is not a formally defined word, but refers to a number of properties all seeking to translate the idea that orbits behave in an unpredictable fashion. The most prominent one is \emph{entropy}, which comes in a number of versions, of which we shall only give one flavor (more information is available in numerous books, e.g. \cite{katok1995modern}). It is to be understood that ``chaotic'' situations are those of positive entropy, zero entropy being the sign for a relatively ``tame'' system.

While we will not \emph{prove} chaos here, we will show compelling numerical evidence that some entropy is positive, meaning that an orbit have the same level of complexity as a random sequence of (skewed) coin tosses. Note that the fact that the map $\Phi$ is chaotic is not the most important point: what matters is that the sequence of winners along many orbits are chaotic, which here results from the attractor intersecting regions electing different candidates.

\begin{figure}[htp]
\centering
\begin{minipage}{.6\linewidth}
\includegraphics[width=\linewidth]{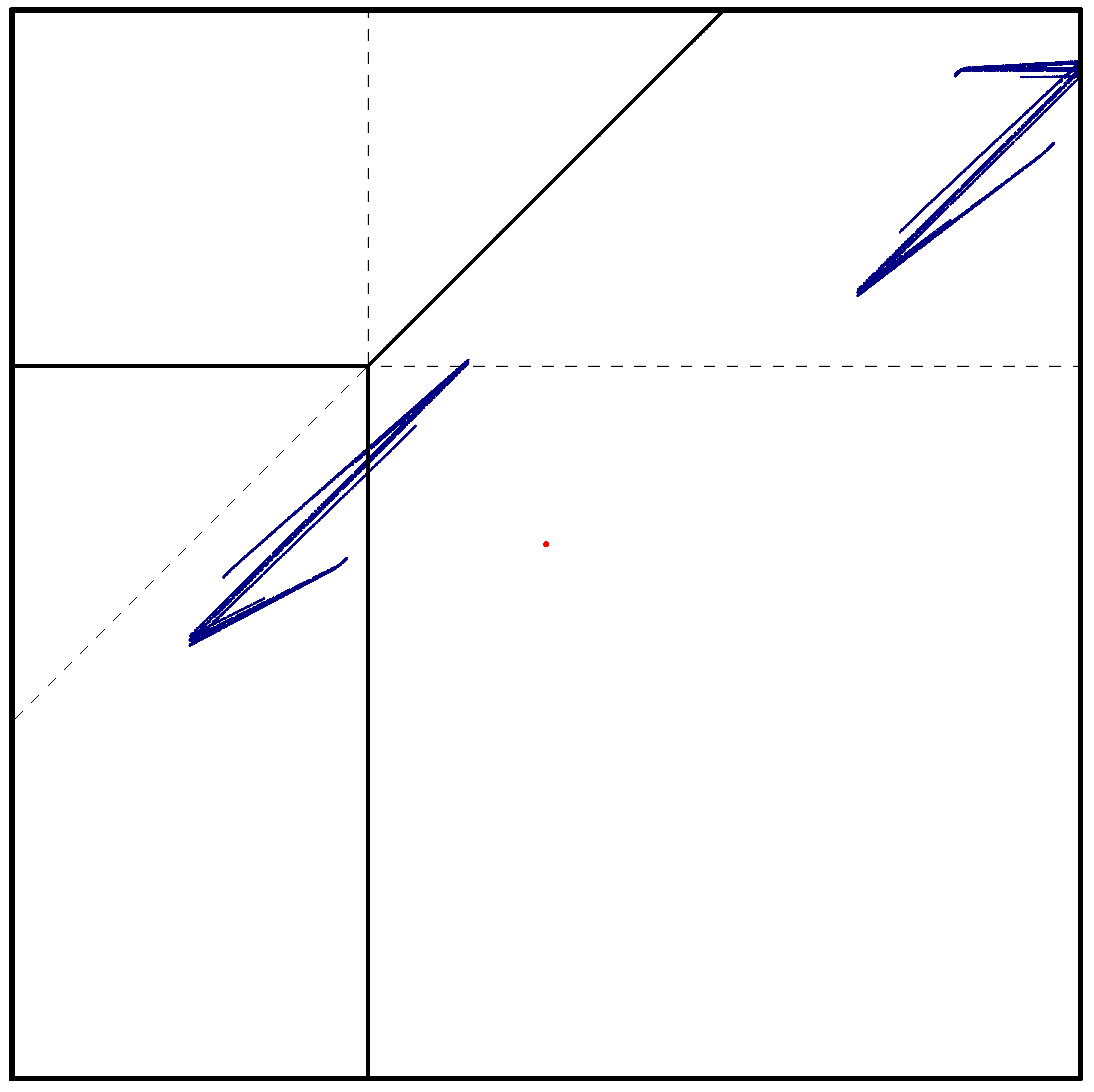}
\end{minipage}
\begin{minipage}{.38\linewidth}
\includegraphics[width=\linewidth]{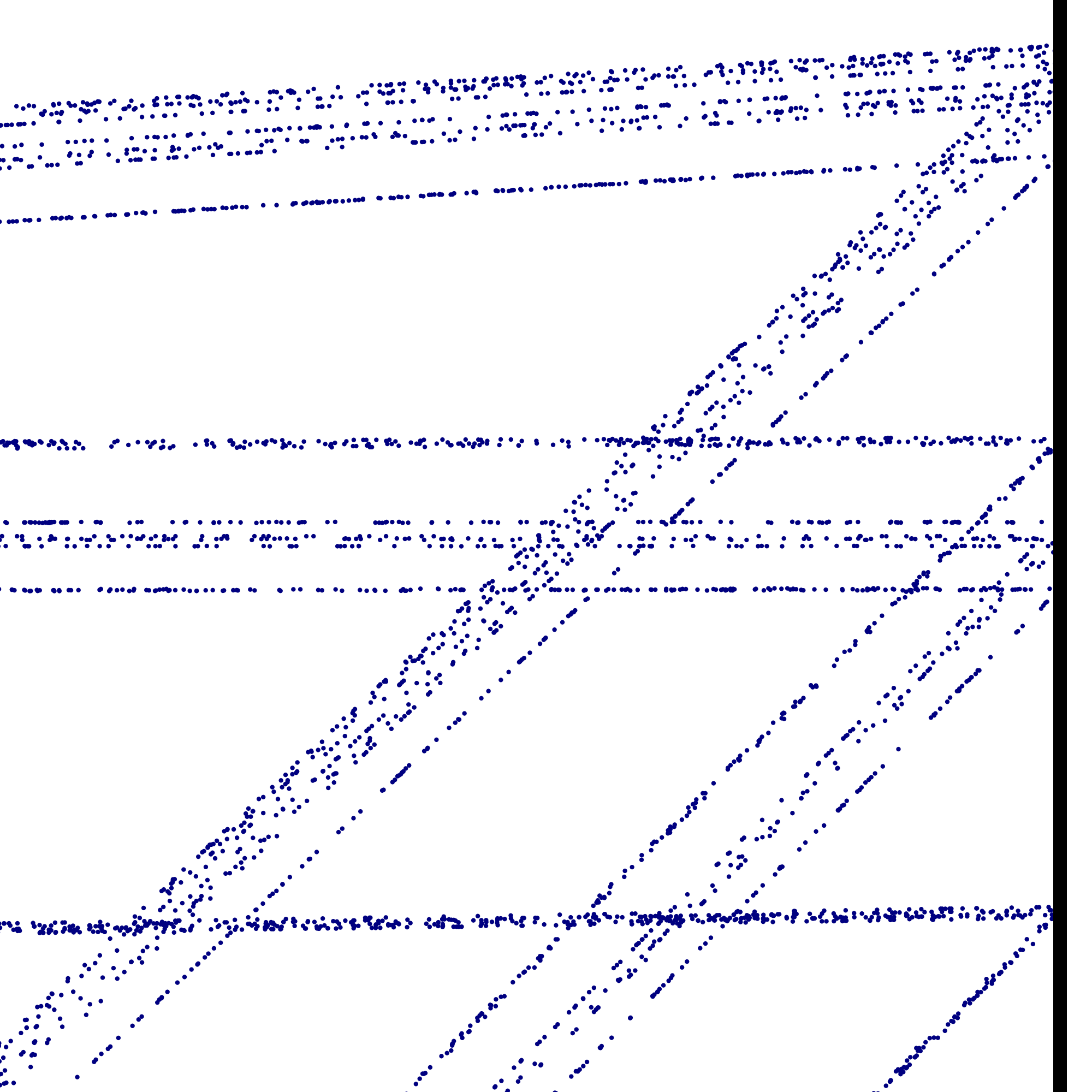}
\end{minipage}
\caption{One orbit of the CS Polling Dynamics $\Phi$. Left: full picture. Right: a zoom to the upper-right part showing a linear/fractal struture.}\label{f:chaos}
\end{figure}

\paragraph{Numerical quantification of the chaos}

Given any state $s\in\mathcal{P}$, its \emph{winners word} is the word $\omega(s)$ whose $k$-th letter is $W(G(\Phi^k(s)))$, the winner after $k$ iterations. In the example above the word of the successive winners starting from the center of the space $\mathcal{P}$ is ($32$ first letters shown, capital $c$ for better readability):
\[\omega(.5,.5) = aaaCaCaaaCaCaCaaaCaCaCaCaCaaaCaC\cdots\]
(the choice of $(.5,.5)$ is arbitrary, but experimentally one checks that it has no bearing on the results exposed here).
There is some visible structure in this word, for example there are never two successive $c$ no four $a$ in a row; apart from the first three letters, the word seems constructed only with the blocs $\beta=ca$ and $\gamma=caaa$. However, the succession of these blocks looks somewhat random:
\[\omega(.5,.5) = aaa\gamma\beta\gamma\gamma\beta\gamma\gamma\gamma\gamma\beta\gamma\gamma\beta\gamma\gamma\gamma\gamma\beta\beta\beta\beta\gamma\dots \]
To quantify the randomness of $\omega(.5,.5)$, we use the following definition.

\begin{defi}\label{d:entropy}
Let $\omega=\alpha_1\alpha_2\cdots$ be an infinite word on the alphabet $A$, i.e. a sequence of elements of $A$ (written without parentheses or comma).
For all $k<j\in\mathbb{N}$ we set $\omega_{k: j} = \alpha_k \alpha_{k+1}\alpha_{k+2}\cdots \alpha_{j}$; each $\omega_{k: j}$ is called a \emph{subword} of $\omega$ (of length $j-k+1$).

For all $k,\ell\in\mathbb{N}$ with $\ell\le k$ and all finite word $\theta$ of length $\ell$, we denote by $S_k^\theta(\omega)$ the number of times the subword $\theta$ appears in $\omega_{1: k}$, we set 
\[P_k^\theta(\omega) := \frac{S_k^\theta(\omega)}{k-\ell+1}\]
the proportion of length $\ell$ subwords of $\omega_{1: k}$ equal to $\theta$, and by $P_k^\ell(\omega)$ the ``probability vector'' $(P_k^\theta(\omega))_{\theta\in A^\ell}$. Finally the \emph{Kolmogorov-Sinai entropy} of the word $\omega$ is defined as
\[h_\mathrm{KS}(\omega) = \lim_{\ell\to\infty} \frac1\ell  \limsup_{k\to\infty} H(P^\ell_k(\omega))\]
where
\[H(p_1,p_2,\dots,p_k) = \sum_{i=1}^k -p_i \log p_i  \qquad \text{with the convention } 0\log 0= 0. \]
\end{defi}

\begin{rema}
The limit in the definition exists by Fekete's lemma, the needed  subadditivity being of the perks of the function $H$. The entropy $h_\mathrm{KS}(\omega)$ is never larger than $\log m$, with equality when all possible subwords of length $\ell$ appear in $\omega$ with the same asymptotic frequency.

The value of $h_\mathrm{KS}(\omega)$ is to be interpreted as the ``uncertainty'' of a random guess. The reference case is that of an object uniformly drawn among $N$, having uncertainty $\log N$; Observe that this makes uncertainty linear in the number of objects to be guessed in the sense that guessing $\ell$ independent objects each uniformly drawn among $N$ has uncertainty $\log(N^\ell) = \ell\log N$. More generally, one defines the uncertainty of a choice made according to a probability vector $(p_1,\dots, p_N)$ as $H(p_1,\dots, p_N)$, and this definition enjoys many natural properties (see e.g. \cite{walters1982ergodic} Theorem 4.1). The uncertainty of guessing a length-$\ell$ subword of $\omega$, when $\ell$ is large and we look far on the right of the word, is asymptotically of the order of $\ell\cdot h_\mathrm{KS}(\omega)$.
\end{rema}

The sequence $\ell\mapsto H(P^\ell_{2^{20}}(\omega))$ is plotted in Figure \ref{f:entropy}. We observe an extremely good alignment from $\ell=4$ onward, with slope $\simeq 0.229$, a strong indication that $h_\mathrm{KS}(\omega(.5,.5))$ is close to this value. Every starting state $s$ yields very similar results, a strong numerical indication that the map $\Phi$ is chaotic, with chaotic sequences of winners. To check whether all states have orbits accumulating on the attractor pictured in Figure \ref{f:continuous-numeric}, we have drawn the first iterates of $\Phi$ in Figure \ref{f:iterations}. All orbits seem to accumulate to the attractor, which appears in fact connected (a fact that can be confirmed by topological arguments), made of a very thin curve between two arrow-shaped  parts, which are exchanged by $\Phi$.

\begin{figure}[htp]
\centering
\includegraphics[width=.6\linewidth]{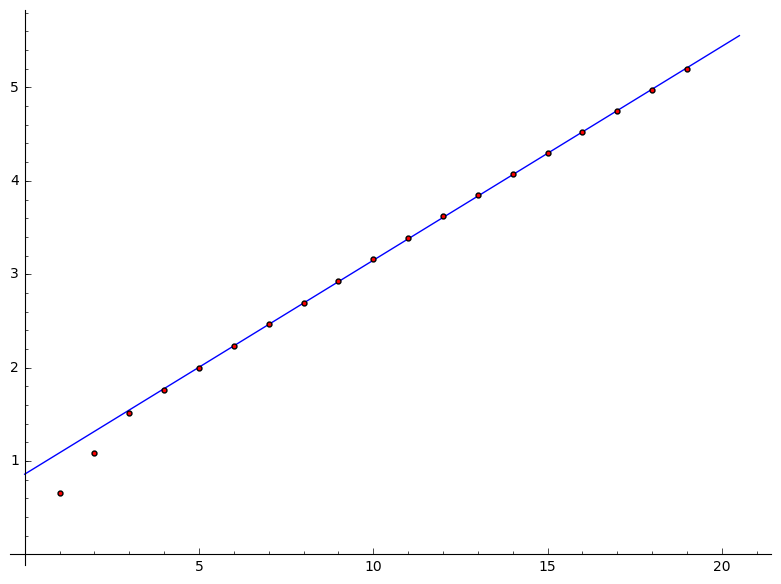}
\caption{Plot of the entropy-estimating sequence  $\ell\mapsto H(P^\ell_{2^{20}}(\omega))$. The slope of the marked line is $\simeq 0.2291$}\label{f:entropy}
\end{figure}

\begin{figure}[htp]
\centering
\includegraphics[width=.24\linewidth]{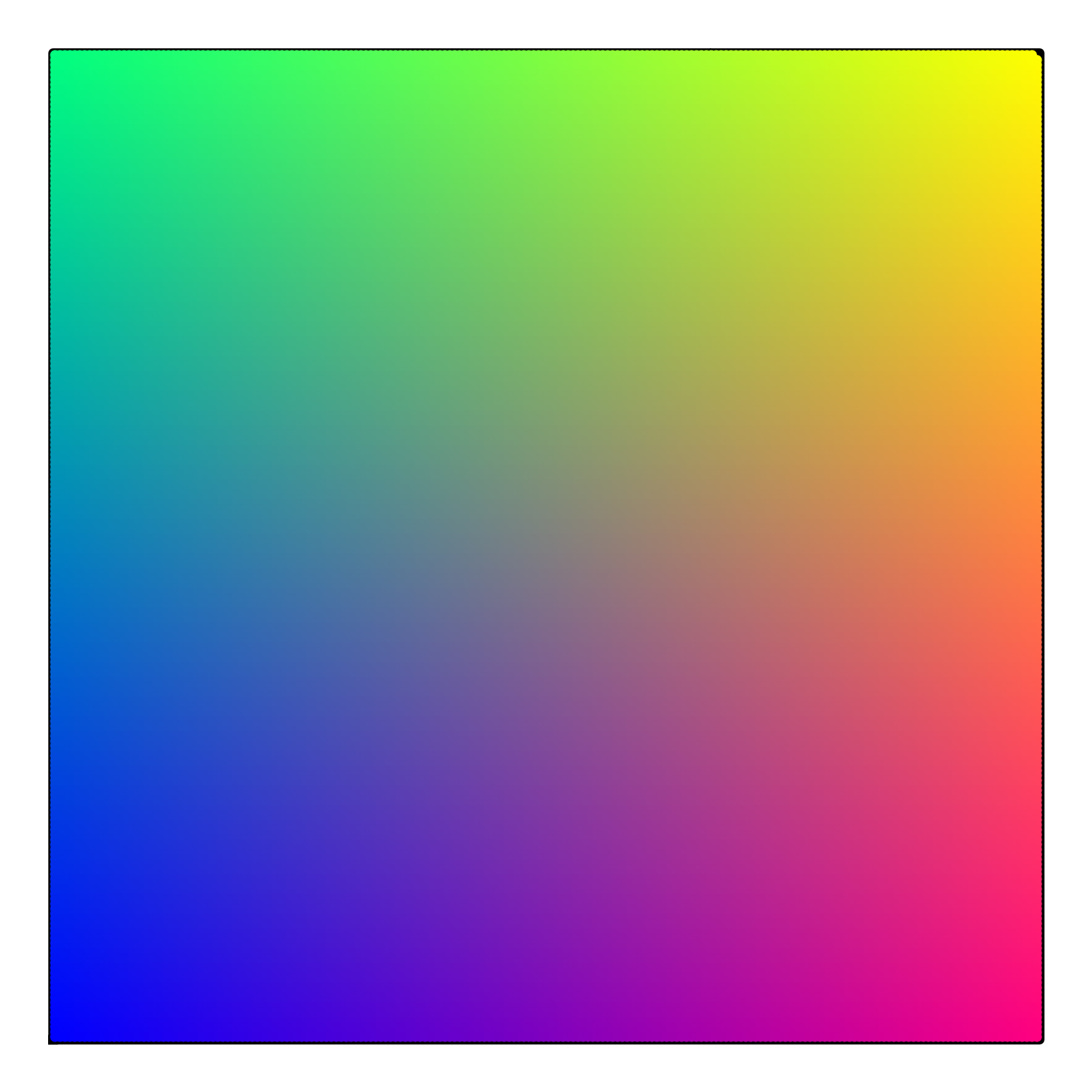}
\includegraphics[width=.24\linewidth]{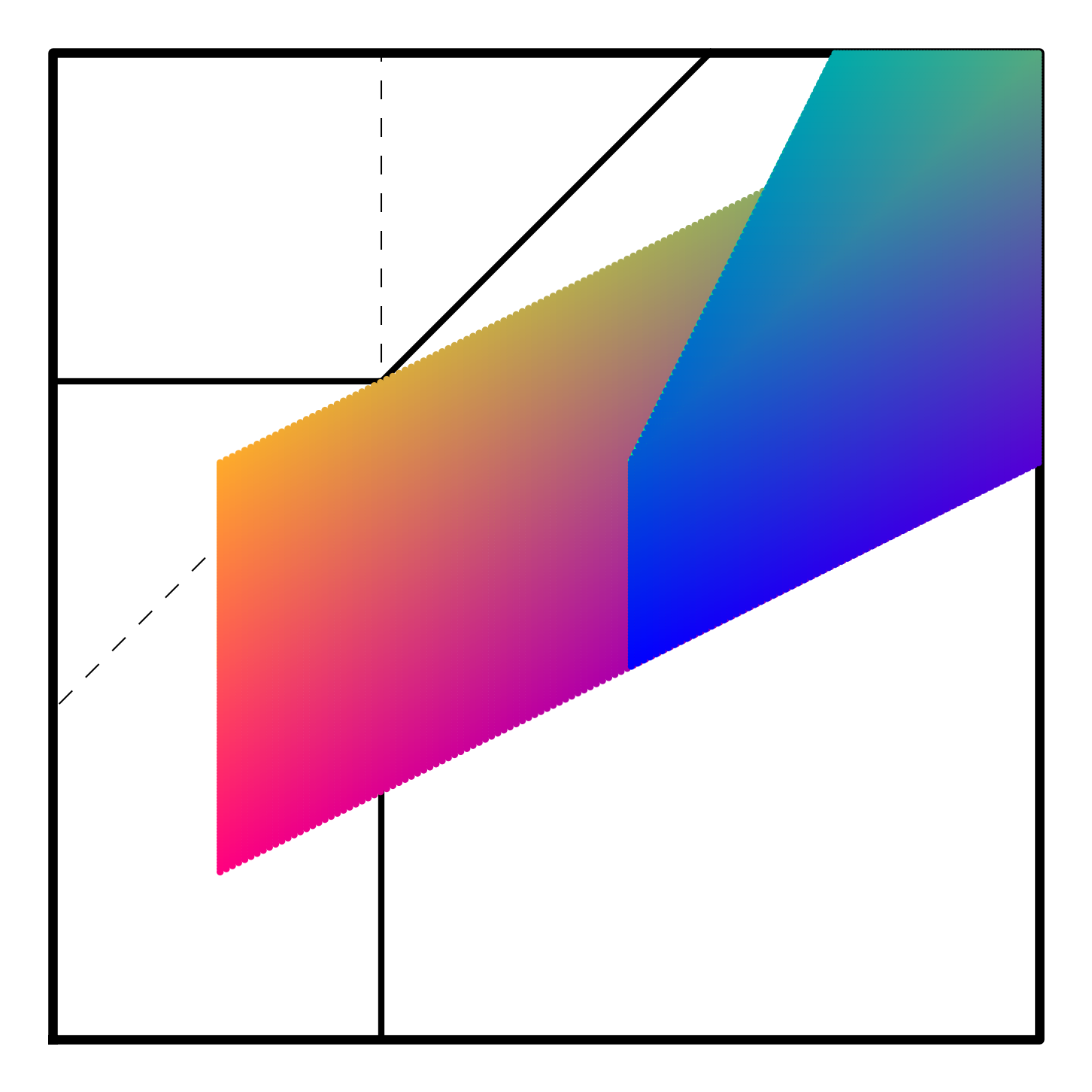}
\includegraphics[width=.24\linewidth]{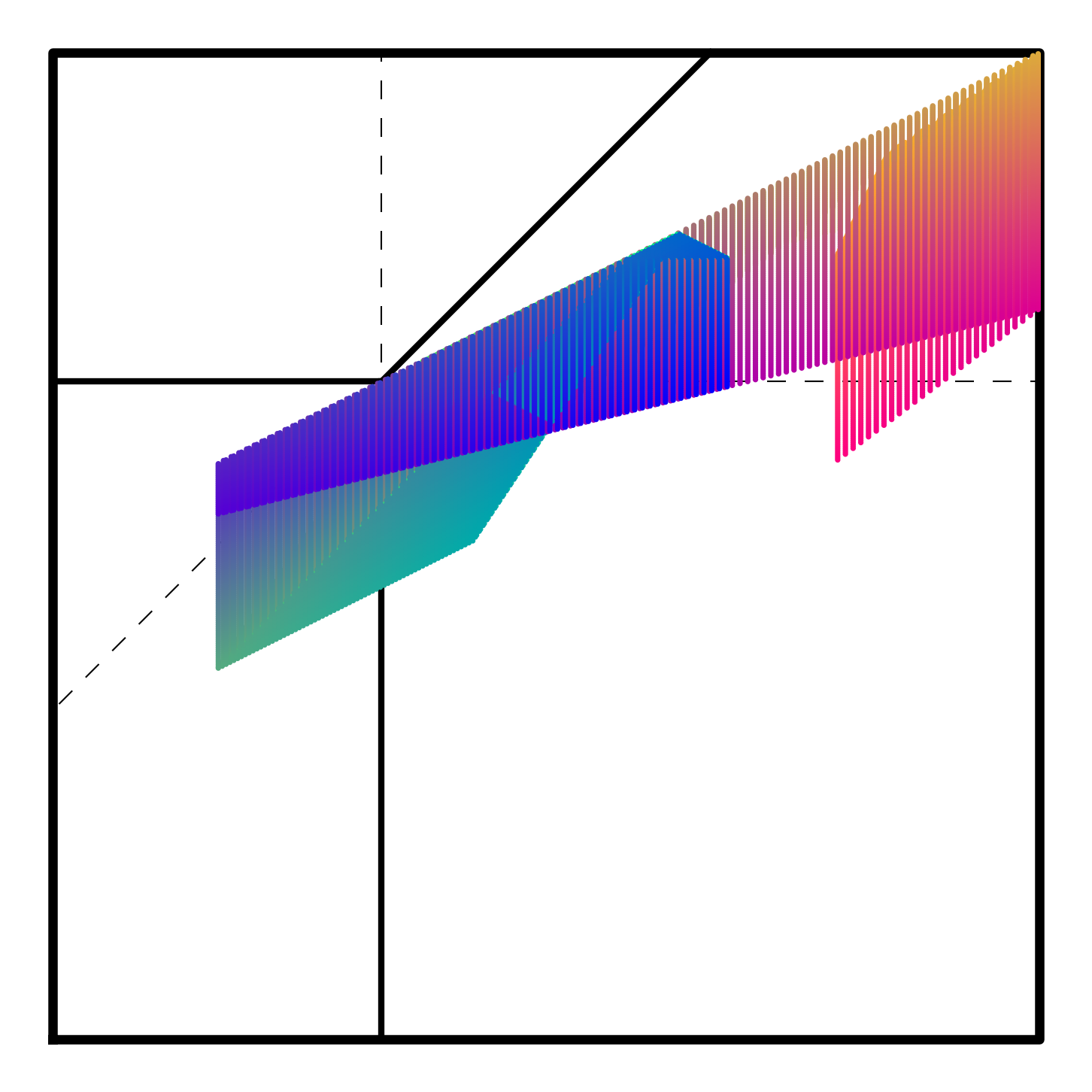}
\includegraphics[width=.24\linewidth]{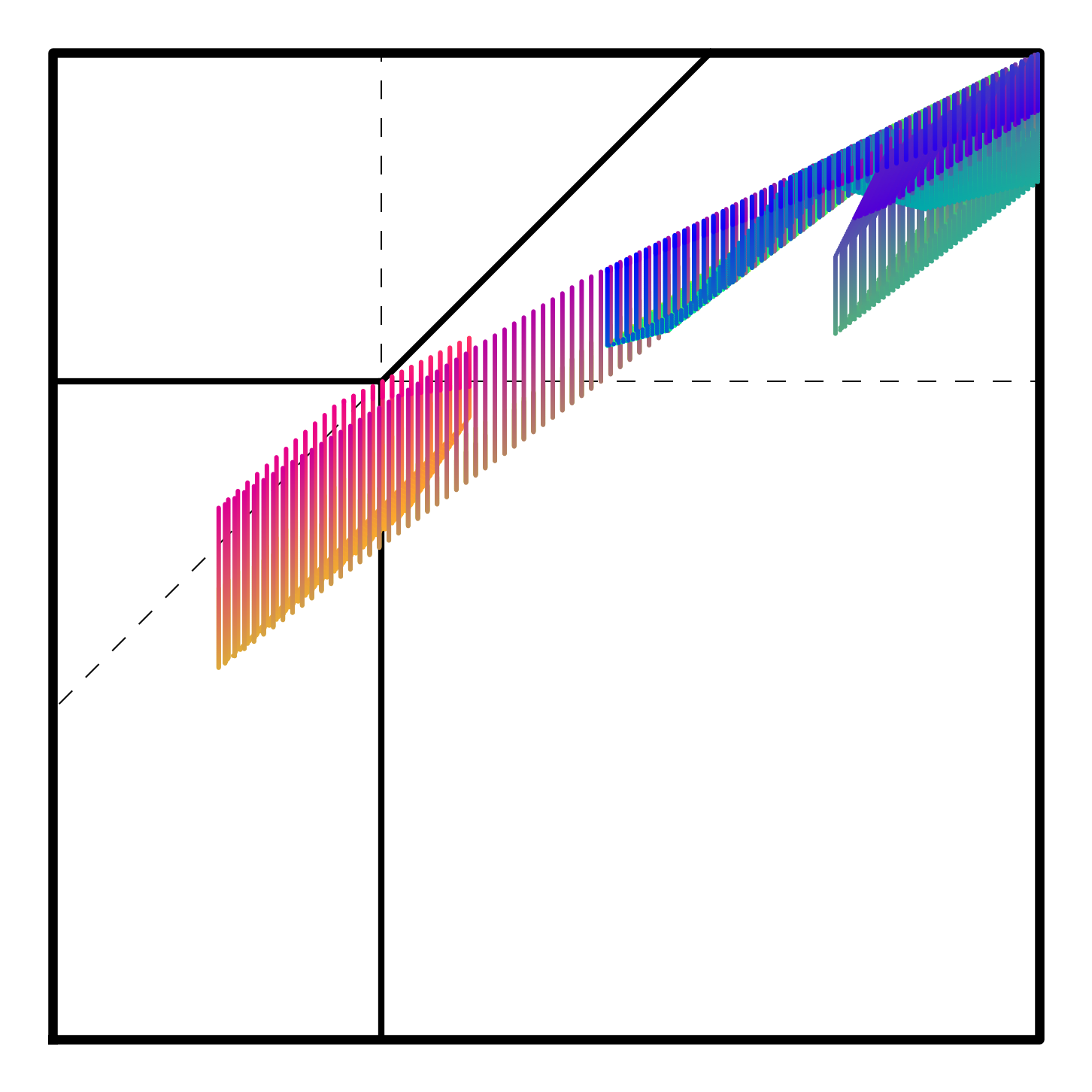}
\includegraphics[width=.24\linewidth]{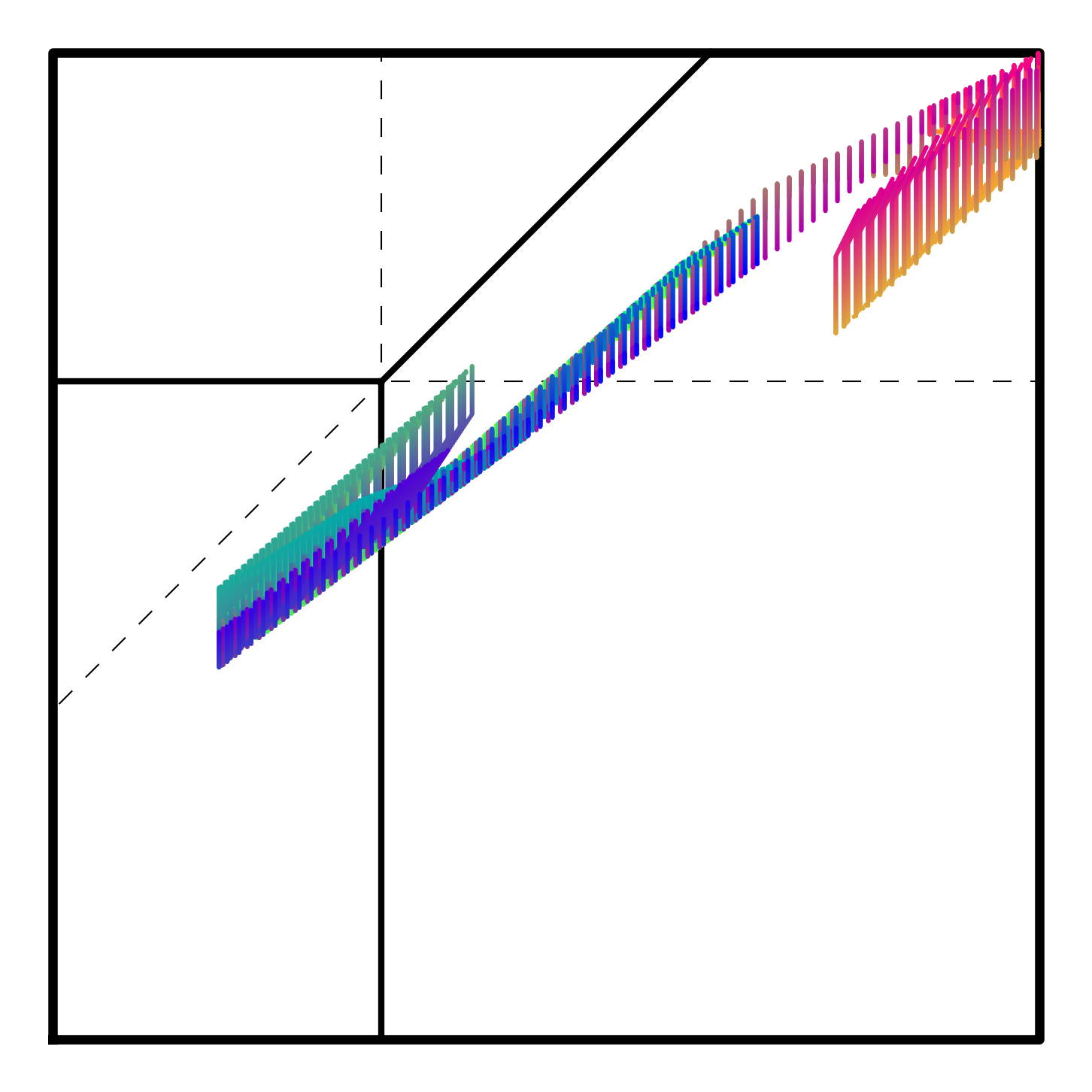}
\includegraphics[width=.24\linewidth]{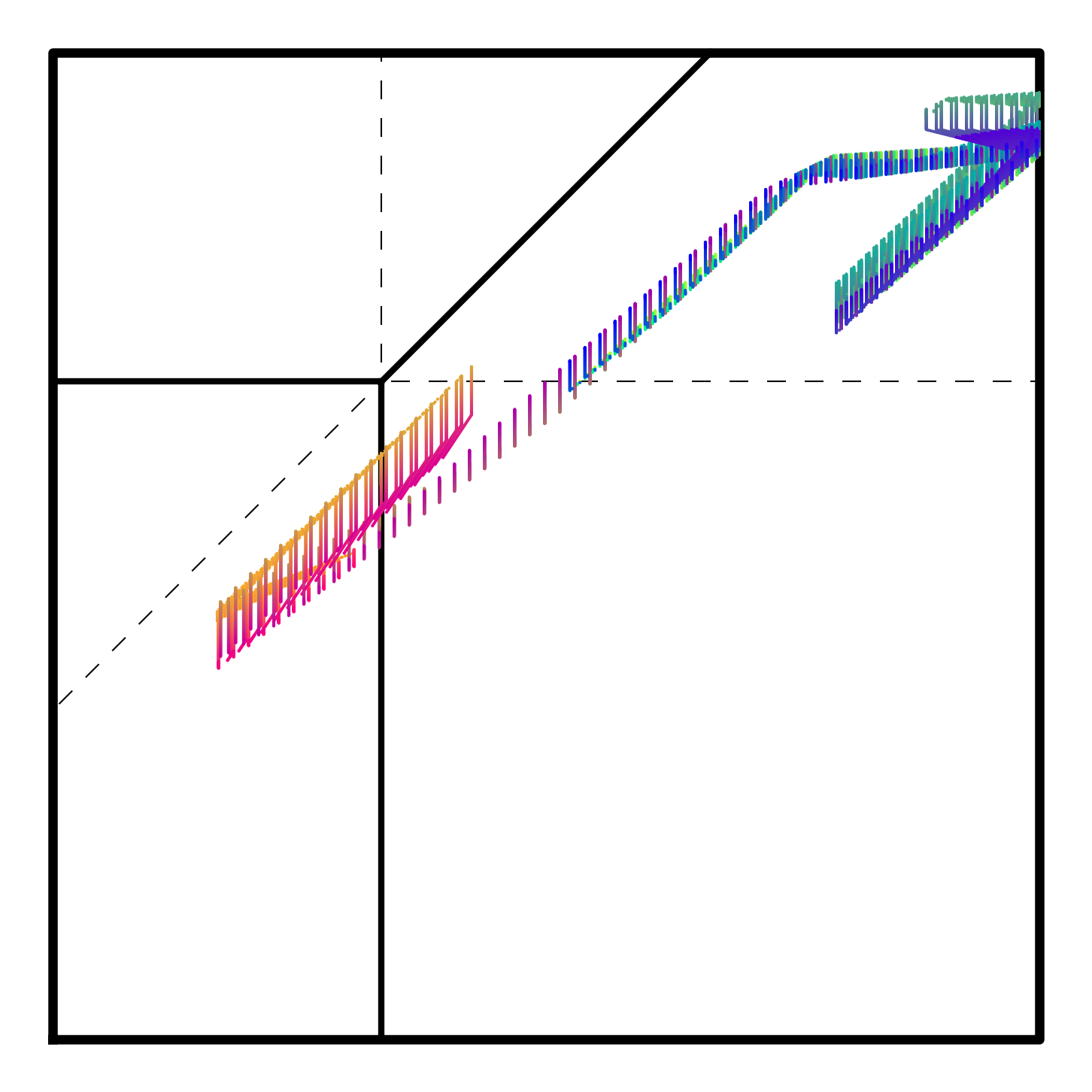}
\includegraphics[width=.24\linewidth]{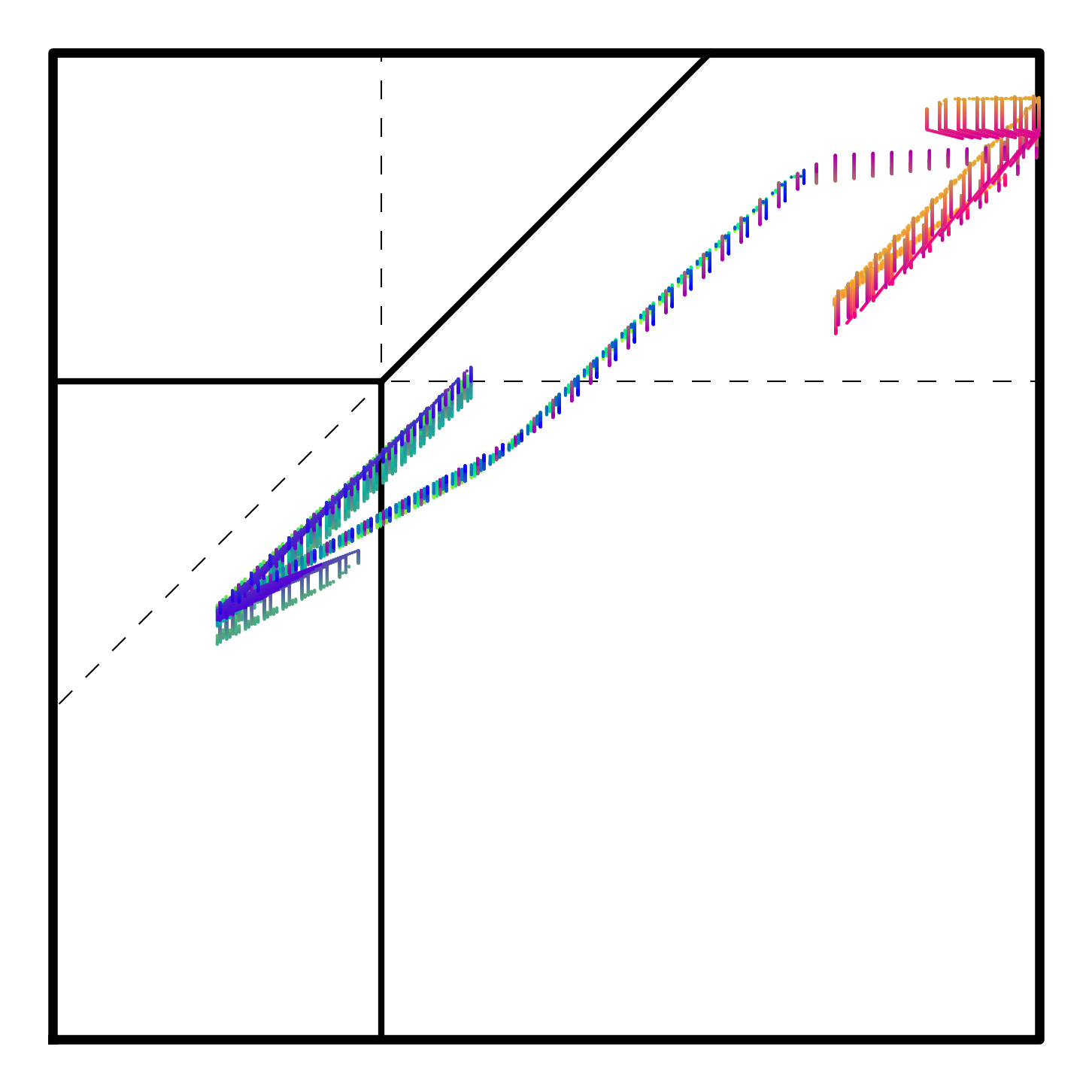}
\includegraphics[width=.24\linewidth]{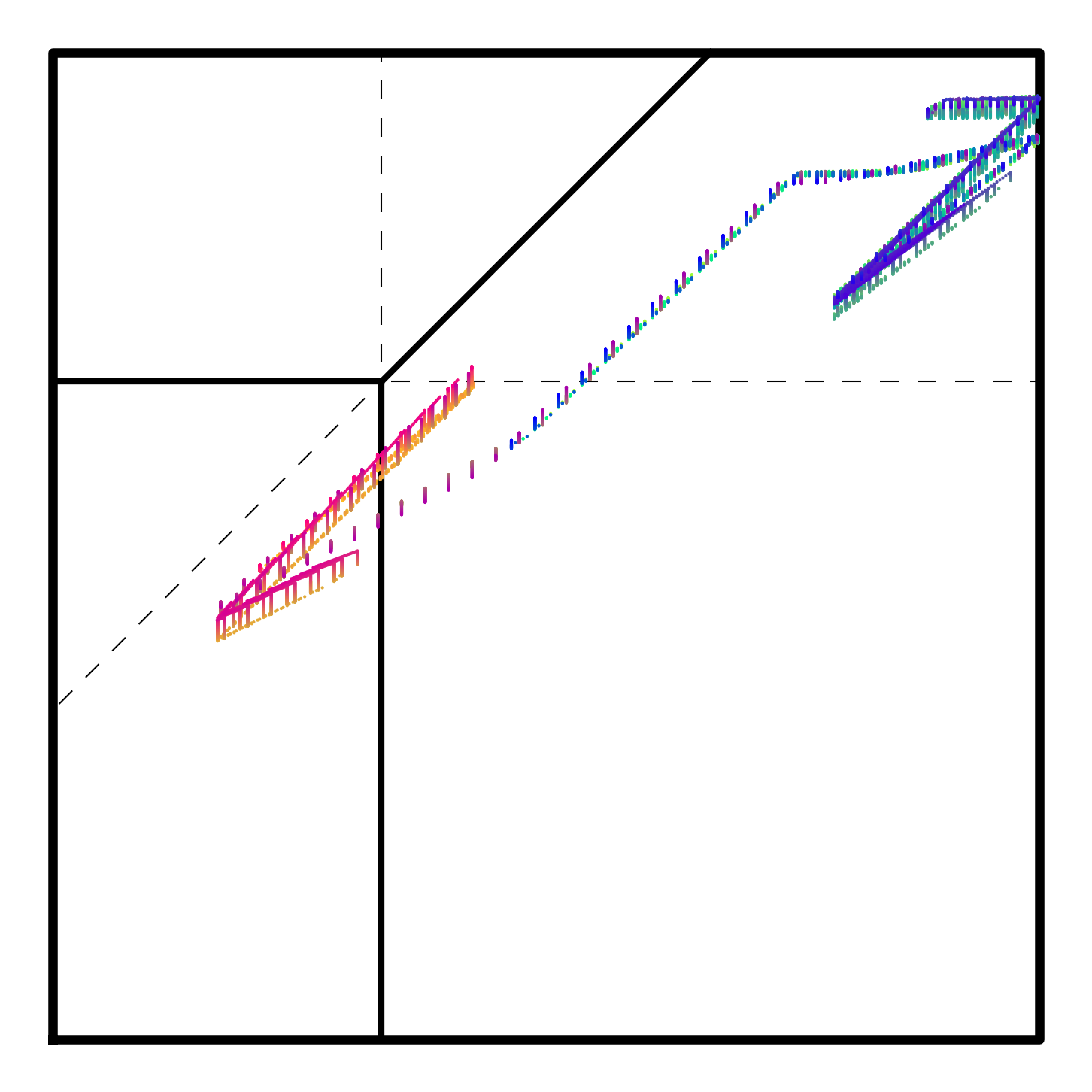}
\caption{From left to right then top to bottom, $200\times 200$ points of the state space regularly spaced and their $7$ first iterates by $\Phi$ (colors preserved along each orbit).}\label{f:iterations}
\end{figure}

While this is only a specific example, it was constructed from basic principles and it is possible that chaos be a quite common feature of CS Polling Dynamics. In order to start testing this hypothesis, we look at what happens when we let $c_X$ and $c_Z$ vary, see Figure \ref{f:entropy-param}. We observe that entropy is mostly concentrated in an L shaped region, where both voter types have not too high a risk aversion, and either type has not too high risk tolerance. However when $Z$ has high risk-aversion, moderate-high values of $c_X$ still result in positive entropy. This asymmetry must comes from the asymmetry of the preference profile introduced by voters of type $Y$.

Looking at sequences of winners with various parameters show that in the upper-right square where entropy vanishes (high risk tolerance for both voter types), several different dynamics can happen: $c$ can be constantly elected or we can observe a periodic pattern involving $a$ and $c$ as winners. When $X$ has high risk aversion, we only observed constant election of $a$. When $Z$ has high risk aversion, for increasing values of $c_X$, we observe successively: constant election of $a$, periodic patterns involving $a$ and $b$, chaotic patterns involving $a$ and $b$, constant election of $b$.\\

\begin{figure}[htp]
\centering
\includegraphics[scale=1]{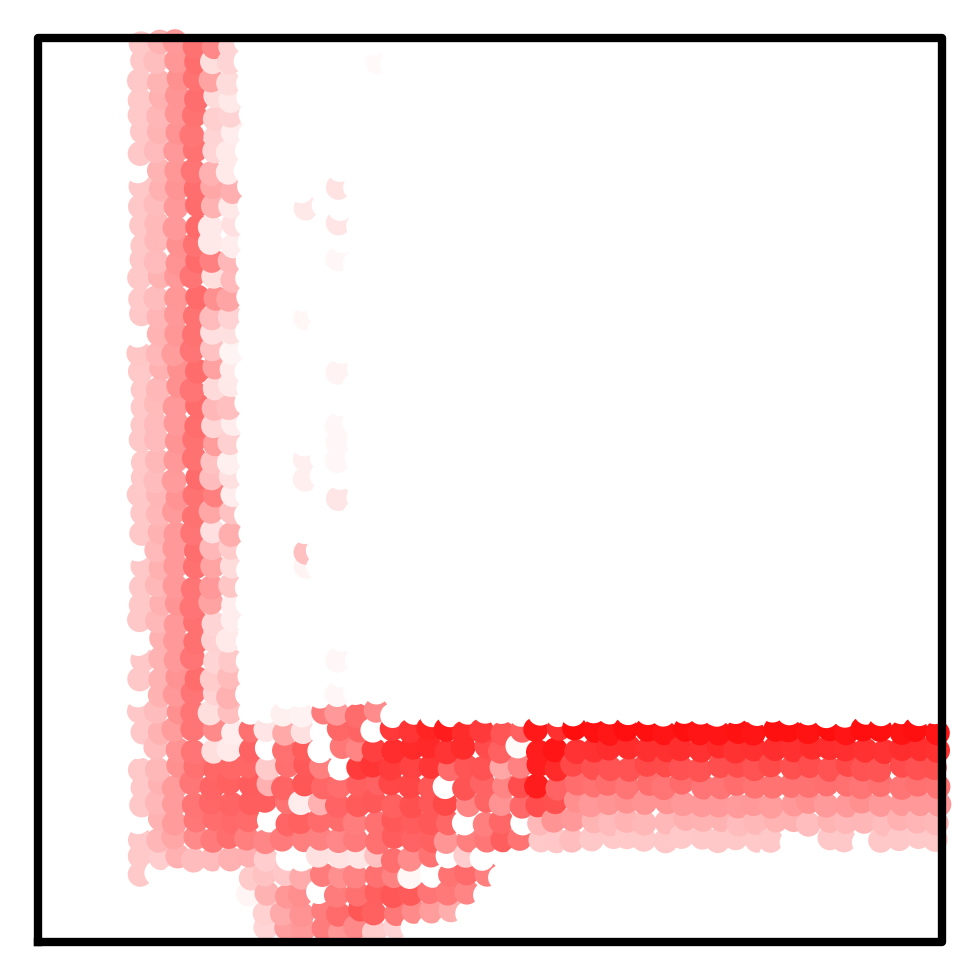}
\caption{Entropy estimated for various values of $c_X$ (abscissa) and $c_Z$ (ordinate), from $1.$ to $30.$, mostly regularly spaced with a random perturbation. Each dot is one value of the pair $(c_X,c_Z)$, color indicating the estimated entropy (linear regression made on subwords of length $4$ to $11$ in $\omega(.5,.5)$): white is entropy $0$ while pure red is entropy $\ln(2)$ or above.}\label{f:entropy-param}
\end{figure}

Let us end with the influence of the preference profile. We look at the basic example ($c_X=c_Z=5.$) with randomly chosen weights for all four voter types. The following examples complements the main one to form a representative sample of what a quick experimental exploration revealed for this model.

\begin{exem}
With $w_X=0.05$, $w_Y=0.02$, $w_Z=0.8$ and $w_W=0.8$, $a$ is constantly elected. Observe that it always gets at least $.82$ votes while $c$ receives only $.8$, so the only possible contender would be $b$, but she would need almost all of $Z$ voters to cooperate and most $X$ voters not to cooperate, a behavior prevented in our model. This is a non-chaotic, actually very stable example.
\end{exem}

\begin{exem}\label{e:intermittent}
Take now $w_X=0.6$, $w_Y=0.08$, $w_Z=0.56$ and $w_W=0.82$.
Then applying the same method as above, we get a higher entropy estimation:
\[h_\mathrm{KS}(\omega(.5,.5)) \simeq 0.36\]
but the word itself looks a bit different, with long sequences of consecutive $a$:
\[\omega(.5,.5) = aaaCaaaCaCaaaaaCaCaaaaaaaaaCaaaaaaaaaaaaaCaCaCaCaCa\cdots\]
This could be example of an ``intermittent'' behavior, with slow regions where $\Phi$ is relatively tame but from which all orbits eventually escape to enter strongly chaotic regions, constantly alternating between two behaviors: predictable and chaotic. See one orbit in Figure \ref{f:examples} (left); other orbits produce very similar pictures.
\end{exem}

\begin{exem}\label{e:no-entropy}
Changing only slightly the previous example with $w_X=0.6$, $w_Y=0.08$, $w_Z=0.56$ and $w_W=0.81$, a radical change in long-term behavior appears: entropy seems to vanish, with $H(P^\ell_{2^{20}})$ plateauing abruptly from $\ell=10$ onward. When looking at subwords of $\omega(.5,.5)$ of length $10$ and more, one observes that there are exactly $22$ of them, no matter the length. Closer observation then reveals that $\omega(.5,.5)$ is the concatenation of copies of the length $22$ word $aaacacaaacacacacaaaaaa$. Other starting states yield similar results, indicating that $\Phi$ has here an attracting periodic orbit, of period $22$ (see Figure \ref{f:examples}, right). At the time scale of an election, it would appear chaotic, but in the longer run it is not. 

This examples incites us to plot very long orbits on top of our entropy estimation, to rule out an attracting periodic orbit of length significantly larger than the maximal length of subwords used in the entropy estimation. Rigorously proving positive entropy would need more sophisticated mathematical tools (one would search for a ``horseshoe'').
\end{exem}

\begin{figure}[htp]
\centering
\includegraphics[width=.48\linewidth]{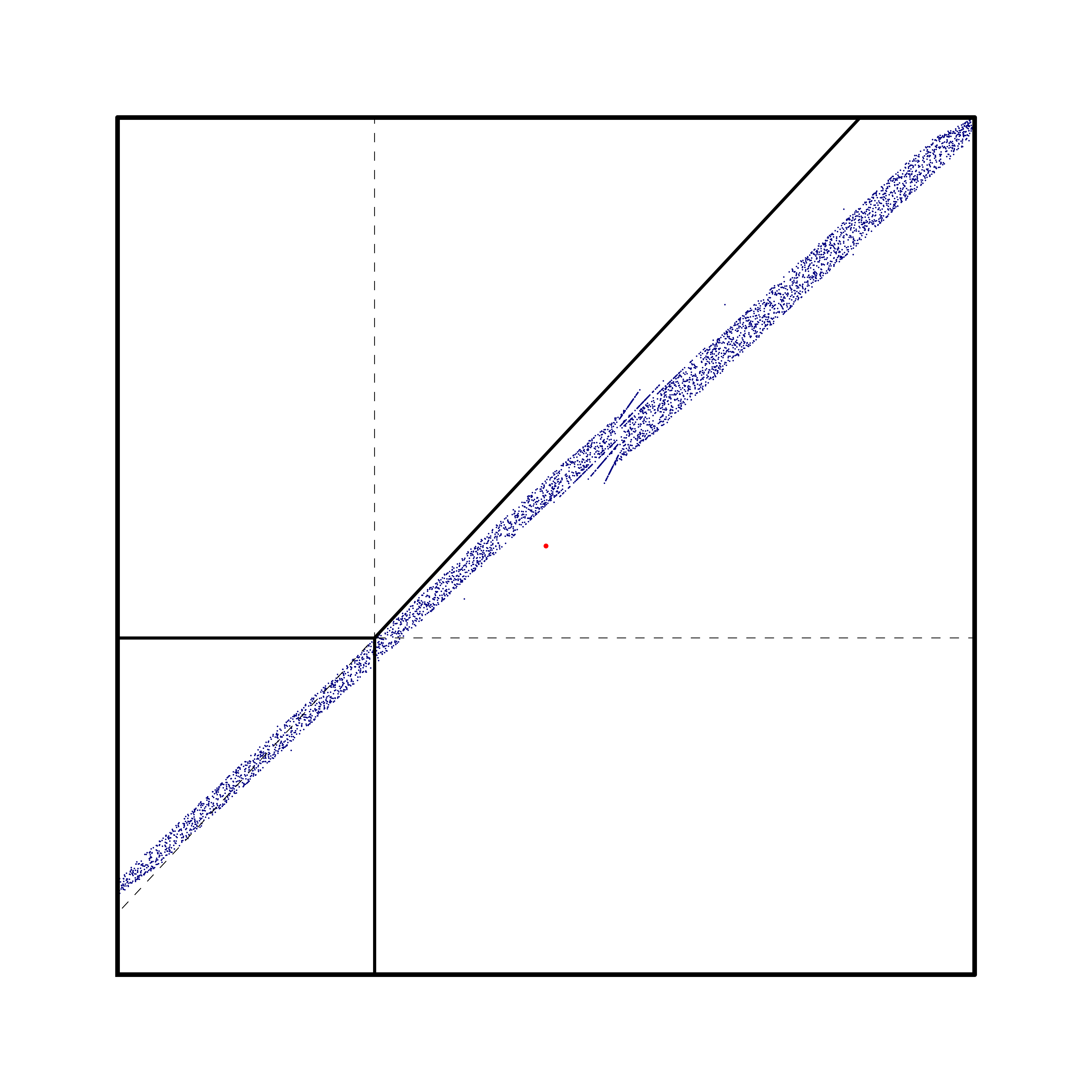}
\includegraphics[width=.48\linewidth]{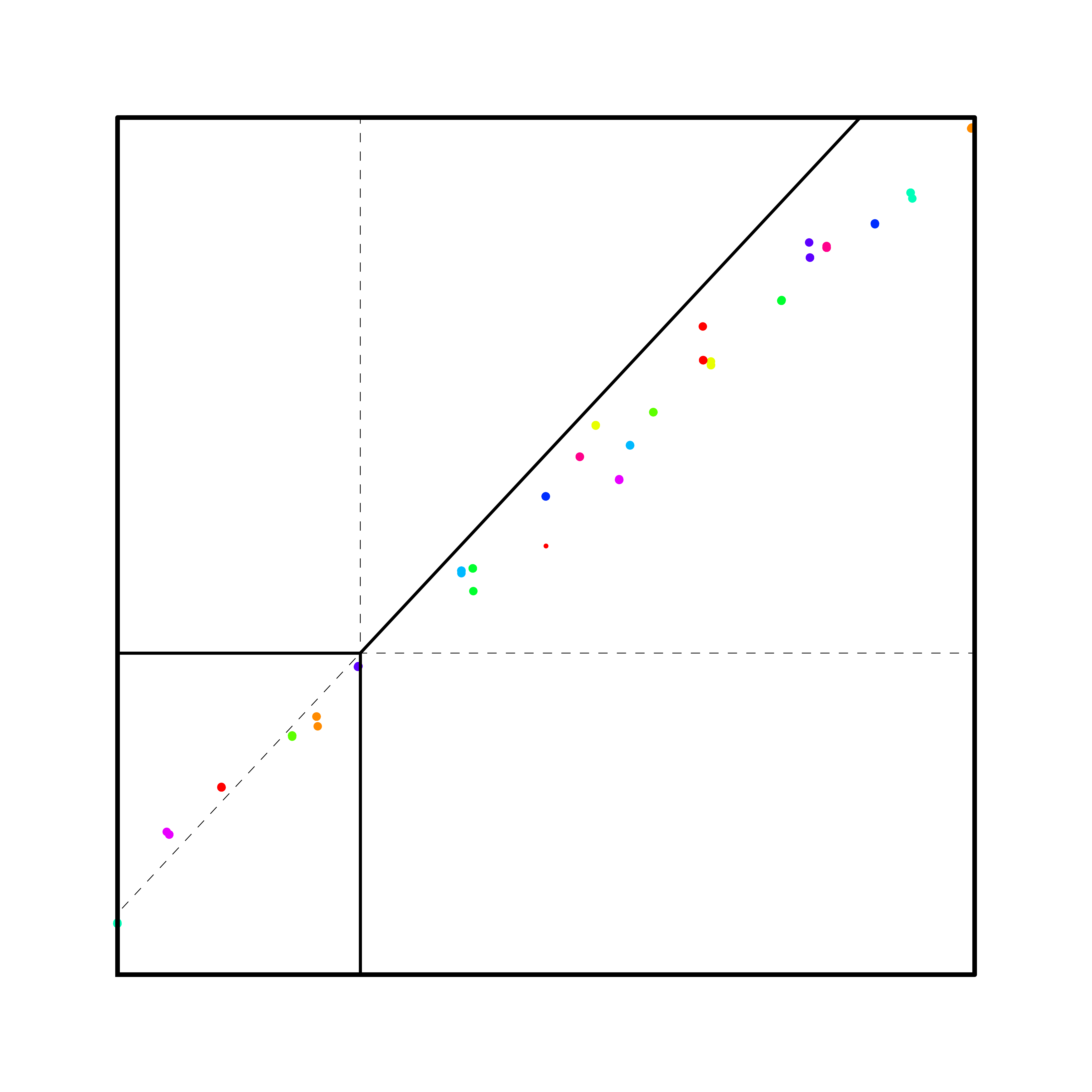}
\caption{Comparison of the orbits of $s=(.5,.5)$ for two close populations: Example \ref{e:intermittent} on the left, Example \ref{e:no-entropy} on the right, with colors of period $22$. Both runs have $5000$ points (larger points are used on the right for readability).}\label{f:examples}
\end{figure}

\begin{figure}[htp]
\centering
\includegraphics[width=.8\linewidth]{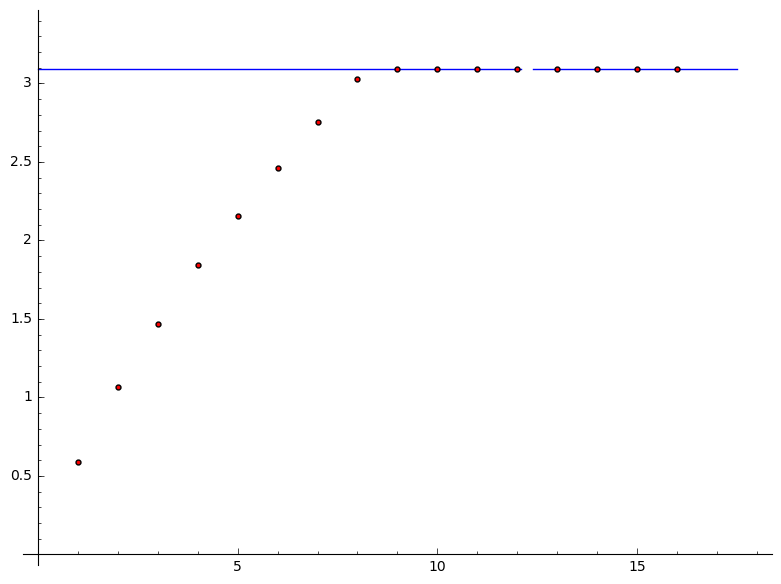}
\caption{Plot of the entropy-estimating sequence  $\ell\mapsto H(P^\ell_{2^{20}}(\omega))$ for Example \ref{e:no-entropy}}\label{f:no-entropy}
\end{figure}

\section{Conclusion}\label{s:conclusion}

We have explained why \emph{synchronized} iterated voting is relevant both to model political elections preceded by polls, and to iterative anticipation of other voters' strategies. We proposed a framework for synchronized iterated voting, producing a ``Polling Dynamics'', in two flavors: one with a discrete space of states, the other with a continuous space allowing more flexible modeling. We showed that the former can be seen as a particular case of the latter.

Our first main result is that cycles of outcomes are robust under perturbation, so that finding a cycle in the discrete-space setting ensures any small enough variation of the model in the continuous-space setting still exhibit the same cycle.

Then we produced two examples in Approval voting, where voters apply \emph{simple, consistent, sincere} heuristics but cycles with sub-optimal or outright bad outcomes appear. The first example somewhat mitigate an important result of Laslier: while the equilibria of the Polling Dynamics following from his Leader Rule, as he proved, elect the Condorcet winner when she exist, they may fail to attract most of the possible outcomes.

We performed \emph{in silico} experiments to assert the prevalence of these electoral conundrums. They are very rare for the Leader Rule, and even impossible when the culture is unidimensional. However a slight relaxation of the Leader Rule makes precisely the unidimensional culture the worst one among those tested, with up to more than $15\%$ of electorate producing a bad cycle.

We thus showed that under Approval Voting, not only convergence to equilibrium may not happen, but cycles can lead individually sound heuristics to result collectively in the worst possible outcome.
We then gave example of bad cycles for other voting systems, in particular Condorcet systems, showing that these issues are not at all specific to Approval Voting.

Last, we considered a simplistic example of continuous-space Polling Dynamics ensuring continuity, i.e. small changes in the expected outcome leads to small changes in the ballots cast. It turned out this model has a chaotic behavior, and we conjectured that chaos is not uncommon at all for continuous-space Polling Dynamics. We supported this conjecture by looking at the influence of some parameters of the model, showing that chaos can be observed in a non-negligible range. The way in which a particular model of voters behavior and a particular preference profile result in either constant, periodic or chaotic patterns is only illustrated here, and would deserve a full study that is way beyond the scope of the present work.

\bibliographystyle{amsalpha}
\bibliography{biblio}
\end{document}